\documentclass{article}
\usepackage{graphicx}

\usepackage{braket}

\title{Log-concavity and tunneling:\\ adiabatic quantum optimization for convex functions (with a spike)}

\author{
Arthur Braida$^{1}$, Elie Bermot$^{1,2}$, Simon Apers$^{1}$\\
\\
$^{1}$ Université Paris Cité, CNRS, IRIF, Paris\\
$^{2}$ Pasqal, 24 Av. Emile Baudot, 91120 Palaiseau, France
}

\date{\today}

\usepackage[left=2.5cm, right=2.5cm, top=2.5cm, bottom=2.5cm]{geometry}
\usepackage{amssymb,amsthm,amsmath}
\usepackage{caption}
\usepackage{subcaption}
\usepackage{commath}
\usepackage{braket}
\usepackage{xcolor}
\usepackage{graphicx}
\usepackage{dsfont}
\usepackage{tikz}
\usepackage{appendix}
\usepackage{tkz-tab}
\usepackage[hidelinks, final=true,colorlinks = true ,allcolors = {blue}]{hyperref}
\usepackage[capitalize]{cleveref}

\usepackage{thm-restate}
\newtheorem{theorem}{Theorem}
\newtheorem{theorem*}{Theorem}
\newtheorem{main_res*}{Main Result}

\newtheorem{lemma}{Lemma}
\newtheorem{corollary}{Corollary}
\newtheorem{corollary*}{Corollary}

\newtheorem{definition}{Definition}
\newtheorem{definition*}{Definition}

\newtheorem{claim}{Claim}

\newtheorem{remark}{Remark}

\begin{document}

\maketitle

\begin{abstract}
Quantum tunneling is expected to provide a computational speedup in quantum computing, a phenomenon that Adiabatic Quantum Optimization (AQO) aims to leverage. While some academic proofs of concept have been studied, such as the “Hamming weight with a spike” (HWS) problem, the algorithmic gains of this effect remain underexplored.
In this work we extend the analysis underlying HWS to more general potentials.

In the first half of the work, we establish (discrete) log-concavity of the ground state as a key structural property in this context.
We devise a framework for establishing log-concavity of the ground state for a large family of discrete, 1-dimensional Schrödinger operators.
The family includes convex potentials, but also certain potentials with local minima.
In the convex case, this provides a discrete version of a continuous result by Brascamp and Lieb \cite{brascamp1976extensions}.
We demonstrate the utility of our result by establishing new spectral gap bounds, going beyond related results by Jarret and Jordan \cite{Jarret_2014} for convex potentials.

In the second half of the work, we use our results on log-concavity to extend the perturbative analysis of HWS by Reichardt \cite{reichardt2004quantum} to the larger family of potentials with log-concave ground state.
As a concrete instantiation, we use our result to extend the HWS analysis from a linear potential (which is exactly solvable) to a quadratic potential (which is no longer solvable).
Our result strongly suggests the broader applicability of tunneling to convex potentials with spikes.
\end{abstract}

\section{Introduction}

Adiabatic quantum optimization (AQO) was introduced as a method for finding the 
minima of complicated cost functions $V:\{0,1\}^n \to \mathbb{R}$ on the hypercube, leveraging quantum fluctuations to escape 
local minima that trap classical algorithms~\cite{finnila1994quantum,Kadowaki_1998,Farhi_2001}. 
Its theoretical analysis reduces to understanding the spectral gap of the 
interpolating Hamiltonian
\begin{equation} \label{eq:AQO-H}
H(s)
= -(1-s) \sum_{i=1}^n X_i + s V
\end{equation}
for $s \in [0,1]$, where $X_i$ is the Pauli-$X$ operator acting on the $i$-th qubit, and $V$ is the diagonal operator satisfying $V \ket{x} = V(x) \ket{x}$ for $x \in \{0,1\}^n$.
If the gap remains at least inverse-polynomially 
large in $n$ throughout the annealing path, then the adiabatic theorem 
guarantees that the target ground state can be prepared efficiently. However, 
establishing lower bounds on the gap is in general a hard problem, and the cases 
where it can be done rigorously are relatively 
few -- see e.g.~\cite{Farhi_2001,roland2002quantum,anderson,reichardt2004quantum,hastings2013obstructions,gilyen2021sub,BCCCMNR_2025unstructured,lee2025quantum} for some examples.

A large part of the results concerns the restrictive yet still interesting family of cost functions that is called ``symmetric'', satisfying $V(x) = V(|x|)$.
A canonical such function is the Hamming weight function $V(x) = |x|$.
Due to the exact solvability of the adiabatic Hamiltonian for this potential, a significant body of work was devoted to the performance of AQO for perturbations of this function~\cite{farhi2002quantum, 
kong2015performancequantumadiabaticalgorithm,Muthukrishnan_2016,reichardt2004quantum,brady2016spectral,bergamaschi2020simulated}.
The resulting ``Hamming weight with a spike'' problem (HWS) concerns potentials of the form $V(x)=|x|+h(x)$, for some perturbation or ``spike'' $h$.
While such a spike can create local minima and obstruct classical methods such as simulated annealing, the aforementioned works derive precise conditions on the shape of the spike under which the gap remains open, and hence AQO still succeeds.
E.g., Reichardt~\cite{reichardt2004quantum} uses a perturbative argument based on the overlap between the unperturbed ground state and the spike to establish a spectral gap bound.
Subsequent 
works~\cite{kong2015performancequantumadiabaticalgorithm, Muthukrishnan_2016,brady2016spectral} 
extended this analysis, some of them using approximation methods such as the spin-coherent 
path-integral instanton approach~\cite{Muthukrishnan_2016} and the discrete 
WKB method~\cite{kong2015performancequantumadiabaticalgorithm}.

These analyses critically hinge on the fact that the unperturbed linear potential $V(x) = |x|$ yields an exactly solvable Hamiltonian, with explicit expressions for the eigenvalues and eigenvectors, and indeed there seems to be no work considering more general convex potentials with a spike.\footnote{There are some exceptions in the \emph{continuous} setting, see e.g.~\cite{liu2023quantum} and references therein.} 
In this work, we go beyond linear potentials and consider a larger family of potentials, including but not limited to the broadly relevant family of convex symmetric potentials.
While exact expressions for the corresponding ground states seem elusive, we establish strong structural properties based on ``log-concavity'' that give us a useful handle on them.
As our main proof-of-concept, we extend the HWS analysis from linear to quadratic potentials (for which already we have no closed-form ground state), showing that the tunneling mechanism persists even in this more challenging setting.
Quadratic potentials are the textbook workhorse for convex optimization, and we see our work as opening up the study of quantum tunneling in the more general setting of discrete convex optimization with AQO.

\subsection{Main contributions}

In the following we summarize our contributions in more detail.

\paragraph{(Ultra-)Log-concavity for a broad class of potentials.}
A discrete sequence $\{p_k\}_{k \in [n]}$ is \emph{log-concave} if for all $k$ it satisfies
\[
p_k^2
\geq p_{k+1}p_{k-1}.
\]
A distribution $\{p_k\}$ is \emph{ultra-log-concave} (of order $n$) if the sequence $\{p_k/\binom{n}{k}\}_{k \in [n]}$ is log-concave.

As part of our first contribution, we introduce log-concavity as a critical tool for analyzing (discrete) AQO.
In the continuous setting, log-concavity is well-known to be instrumental to proving spectral gap bounds, as is witnessed e.g.~by all the recent works on log-concave sampling \cite{chewi2022logconcave}.
In the discrete setting however, log-concavity seems to have been less explored in this context.
As an example, a key result by Brascamp and Lieb \cite{brascamp1976extensions} in the continuous setting established that the ground state of a Schrödinger operator with convex potential is (ultra-)log-concave.
Turning to the discrete setting, we were unable to find an analog of this key result.
As our first main result, we establish discrete log-concavity of the ground state of a rather large family of Hamiltonians, including Schrödinger operators on the discrete line with single-well potential, and the AQO Hamiltonian (\cref{eq:AQO-H}) with symmetric potentials that are convex (and even certain multi-well potentials).

We prove all of our results for 1-dimensional Hamiltonians of the form $H = -A + V$, where $A$ is band-diagonal with $A_{k,k-1}=A_{k-1,k}=a_k$ for log-concave $\{a_k\}$.
This includes path Hamiltonians, but also AQO Hamiltonians (cf.~\cref{eq:AQO-H}) with a symmetric potential, for which there is a well-known symmetry reduction from the hypercube to the path. 
We establish a condition called ``relative monotonicity'' (RM) of the potential $V$ with respect to the coefficients $\{a_k\}$ that ensures log-concavity of the ground state.
The RM-condition constrains ratios of consecutive values of the potential compared to the growth of $a_{k+1}/a_k$.

\begin{main_res*} [\cref{thm:ultra log-conc-conv} and \cref{informal:thm_lc_hill}]
\, \newline
$\bullet$ If a 1-dimensional Hamiltonian $H = -A + V$ satisfies the RM-condition then its ground state is \emph{log-concave}.\\
$\bullet$ If $H$ is the symmetry-reduced AQO Hamiltonian with symmetric and convex $V$, then the ground-state is \emph{ultra-log-concave}.
\end{main_res*}

The second bullet yields a discrete version of the continuous result by Brascamp and Lieb \cite{brascamp1976extensions}.
The first bullet allows us to go beyond convexity.
E.g., if the $a_k$'s are uniform (i.e., $H$ is the Schrödinger operator on a path), then the RM-condition merely requires that the potential be single-well (i.e., $-V$ must be unimodal) rather than convex.
On the other hand, in the hypercube case, the RM condition even captures certain multi-well potentials, having local minima.
This demonstrates how log-concavity can already capture certain quantum tunneling effects.

\paragraph{Spectral gap for log-concave ground state.}
As a first demonstration of the utility of log-concavity, we show that log-concavity (resp.~ultra-log-concavity) of the ground state implies an inverse-polynomial (resp.~constant) lower bound on the spectral gap of $H$.
More specifically, we use a discrete Hardy inequality \cite{miclo1999example} to prove the following theorem.

\begin{main_res*}[\cref{thm:lc-gap}]
Let $H = -A+V$ be a 1-dimensional Schrödinger operator with log-concave ground state $\{\psi_k\}$ and spectral gap $\Delta$.
Then
\[
\Delta \in \Omega\left(\frac{1}{n S_n} \right)
\]
with $S_n=\sum_k \frac{1}{a_k}$. In the case of the AQO Hamiltonian (\cref{eq:AQO-H}) with $V$ symmetric and convex we get \[\Delta \ge \frac{1}{4}.\]
\end{main_res*}

\noindent
We note that $S_n = n$ on the line, while $S_n \leq \pi$ for the symmetry-reduced hypercube.

Combined with our characterization of potentials yielding log-concave ground states, this implies polynomial lower bounds for a broad range of (convex, single-well, and even multi-well) potentials.
We compare our results with existing bounds, e.g.~those by Jarret and Jordan on a discrete version of the fundamental gap lemma for convex potentials \cite{Jarret_2014}.
Our bounds either recover their bounds (up to constants), or extend them beyond convex potentials.

\paragraph{Log-concavity enables tunneling through spikes.}
As a second demonstration, we show that the perturbative spike analysis of Reichardt \cite{reichardt2004quantum} can be easily extended from binomial states to general log-concave ground states.
We exploit their structural properties, such as unimodality and concentration, to establish the following result.

\begin{main_res*}[\cref{thm:spike_general}]
Consider a 1-dimensional Hamiltonian $H = -A + V$ with spectral gap $\Delta$ and a log-concave ground state with variance $\sigma$ and mean $\mu$.
Consider a non-negative perturbation $h$ supported on an interval $[\ell,u)$ and satisfying an upper bound $h(k) \leq \bar h$.
If either
\[
\bar h \frac{u-\ell}{\sigma}
 \in O(\Delta)
\quad \text{ or } \quad
\frac{\max\{\ell-\mu,\mu-u\}}{\sigma}
\in \Omega(\Delta),
\]
then the perturbed Hamiltonian with potential $V + h$ has spectral gap $\Omega(\Delta)$.
\end{main_res*}

This result relies on log-concavity to reduce the analysis of spike tunneling to estimating the mean and variance of the unperturbed ground state, and it provides explicit conditions on the spike parameters under which tunneling occurs. 

\paragraph{A tractable case: quadratic potentials.}
As our final demonstration (and our initial motivation), we deploy these results to study the ``quadratic Hamming weight with a spike'' (quadratic HWS) problem as a key proof of concept.
Quadratic HWS provides a direct extension to Reichardt's result on the (linear) HWS problem~\cite{reichardt2004quantum}, and it shows that tunneling happens for the quadratic case under the exact same conditions.

Specifically, we specialize to quadratic potentials $V_q(k) = a(k - k_0)^2$ for $k \in [n]$ and $k_0 \le 0$, so that the minimizer remains at $k=0$.
While for linear potentials, the ground state is known to follow a binomial distribution $\ket{\varphi_m}$ with mean $m$, there seems to be no closed-form expression available for the ground state under a quadratic potential.
Nonetheless, we show the somewhat surprising result that the ground state maintains a constant overlap of at least $1/2$ with a binomial state $\ket{\varphi_m}$.
Combining this with the log-concavity of the ground state that we established, such overlap bound yields tight bounds on the mean and variance of the quadratic ground state, throughout the whole annealing path.
Plugging these estimates into the spike bound above, we obtain our final result.

\begin{main_res*}[\cref{theorem:gap_quadra}]
Consider the AQO Hamiltonian (\cref{eq:AQO-H}) with symmetric potential $V(k) = a(k-k_0)^2$ for $k_0 \leq 0$, and spike perturbation $h$, supported on $[\ell,u) \subseteq [0,n]$ and satisfying $h(k) \leq \bar h$.
If
\begin{equation} \label{eq:spike-cond}
\bar{h}\,\frac{u-\ell}{\sqrt{\ell}}
\in O(1),
\end{equation}
then AQO prepares an approximation of $\ket{\psi_{qs}(1)} = \ket{0}$ from $\ket{\psi_{qs}(0)}$, the binomial state with mean at $n/2$, in time~$O(n^4)$.
\end{main_res*}

The spike condition \cref{eq:spike-cond} is identical to the one derived by Reichardt for the linear HWS problem~\cite{reichardt2004quantum}, confirming that the tunneling mechanism persists for quadratic potentials and strongly suggesting the broader applicability of AQO to convex potentials with spikes.

\subsection{Organization of the paper}

We now give a brief overview on the remainder of the paper.

In \cref{sec:prelim}, we collect preliminary facts on adiabatic quantum optimization and discrete log-concavity.
In \cref{sec:lc_gs} we prove our main results on log-concave ground states.
In \cref{sub:spectral_gap_lc} we use these to establish new spectral gap bounds based on Hardy's inequality.
In \cref{sec:LC-tunnel}, we use log-concavity to analyze quantum tunneling through spike perturbations.
Finally, in \cref{sec:quadra}, we analyze the quadratic HWS problem.

\section{Preliminaries} \label{sec:prelim}
\subsection{Adiabatic quantum optimization} \label{sec:AQO}
In this section, we recall the adiabatic theorem in the form that is relevant for preparing the ground state of a target Hamiltonian~$H(1)$.
Let $H(s)$, $s\in[0,1]$, be a smooth family of Hamiltonians interpolating between an initial Hamiltonian $H(0)$ and a final Hamiltonian $H(1)$.
Assume that the ground state of $H(s)$ is non-degenerate for all $s\in[0,1]$, and let $\Delta(s)$ denote the spectral gap between the ground state and the first excited state of $H(s)$, with
\[
\Delta_{\min} := \min_{s\in[0,1]} \Delta(s).
\]
The quantum adiabatic algorithm starts in the ground state of an easily preparable initial Hamiltonian $H(0)$ and evolves according to Schr\"odinger's equation:
\begin{equation} \label{eq:schr}
i\frac{d}{ds}U(s) = \tau H\!\left(s\right)U(s), \qquad \forall s \in [0, 1], \qquad U(0) = I.
\end{equation}
The quantity $\tau$ measures the effective evolution time of the resulting algorithm.
The adiabatic theorem then guarantees that, provided that this evolution time $\tau$ is sufficiently large, the evolved state
\[
\ket{\phi(s)} = U(s)\ket{\psi_0(0)}
\]
remains close to the instantaneous ground state $\ket{\psi_0(s)}$ throughout the evolution, and in particular to $\ket{\psi_0(1)}$ at the end.

\begin{theorem}[Adiabatic theorem, special case of~\cite{reichardt2004quantum}] \label{adiabatic_theorem}
Let $H(s)$, $s\in[0,1]$, be a twice differentiable family of Hamiltonians.
Assume that the corresponding ground state energy $\lambda_0(s)$ is non-degenerate and can be defined continuously for all $s\in[0,1]$, and that the spectral gap satisfies
\[
\Delta(s) \geq \Delta_{\min} > 0.
\]
If $\|H(s)\|,\|\dot{H}(s)\| \leq h_{\max}$ and
\[
\tau \in \Omega\left(\frac{h_{\max}^2}{\varepsilon \Delta_{\min}^3}\right),
\]
then $|\braket{\psi_0(1)|\phi(1)}|^2 \geq 1-\varepsilon^2$.
\end{theorem}

In this work, we consider only linear interpolations of the form
\[
H(s) = (1-s)H_0 + s H_1,
\]
for which the ground state energy is assumed to remain non-degenerate along the path, and hence it can be chosen continuously in $s$.
The initial Hamiltonian $H_0$ is chosen so that its ground state is easy to prepare, while the final Hamiltonian $H_1$ is the target Hamiltonian whose ground state we wish to obtain.
For a linear interpolation, the Hamiltonian is twice differentiable and
\[
\dot{H}(s) = H_1 - H_0,
\]
so that $\|\dot{H}(s)\| = O(\|H_0\|+\|H_1\|)$.
For constant error $\varepsilon$, the runtime bound therefore reduces to
\[
\tau = O\!\left(\frac{(\|H_0\|+\|H_1\|)^2}{\Delta_{\min}^3}\right).
\]
The main challenge is thus to prove a good lower bound on $\Delta_{\min}$.
Since in general practical settings, we have $\|H_0\|,\|H_1\| \in \mathrm{poly}(n)$, the adiabatic algorithm is efficient provided that
\[
\Delta_{\min} \in \Omega\!\left(\frac{1}{\mathrm{poly}(n)}\right).
\]

As a special case, we often look at
\[
H_0 = -\sum_{i=1}^n X_i,
\]
whose unique ground state is the uniform superposition
\[
\ket{\psi_0(0)}=\ket{+}^{\otimes n}
= \left(\frac{\ket{0}+\ket{1}}{\sqrt{2}}\right)^{\otimes n}.
\]
If, moreover, $H_1$ is diagonal in the computational basis, then the interpolation is stoquastic, and the Perron-Frobenius theorem implies that the ground state may be chosen to have positive real amplitudes.
In this case, one may equivalently write
\[
H(s)=-(1-s)\sum_{i=1}^n X_i + sV,
\]
with $V$ a diagonal potential.

\subsection{Symmetry reduction} \label{sec:sym}

In this section, we describe the reduction from the bitstring space $\{0,1\}^n$ of dimension $2^n$ to the Hamming weight subspace of dimension $n+1$ spanned by 
uniform superpositions over bitstrings of a given Hamming weight.

\begin{definition}
The Hamming weight subspace is spanned by the vectors
\begin{equation*}
\ket{k}=\frac{1}{\sqrt{\binom{n}{k}}}\sum_{|x|=k}\ket{x},
\qquad 0 \leq k \leq n.
\end{equation*}
\end{definition}

The key point is that, for a symmetric potential $V$ satisfying $V\ket{x} = V(|x|)\ket{x}$, the adiabatic algorithm remains restricted to the Hamming weight basis.
In this subspace, the mixing Hamiltonian $\sum_i X_i$ becomes the tridiagonal operator:
\begin{equation} \label{eq:sym-red}
(A_{hc})_{k-1,k}=(A_{hc})_{k,k-1}:=a_k:=\sqrt{k(n+1-k)}
\end{equation}
for $0<k\leq n$.
Since we assume symmetric potentials throughout this work, we will only consider symmetry-reduced Hamiltonians of the form $H(s)=-(1-s)A_{hc}+sV.$
A more detailed treatment of this reduction is given in \cite[Section~IV]{Jarret_2014}.

\subsection{Binomial state and Hamming weight Hamiltonian}\label{ssec:binHW}

We define the binomial state as follows.

\begin{definition}[Binomial state $\ket{\varphi_m}$]\label{def:psi_m}
For $m \in [0,n/2]$, define the binomial state $\ket{\varphi_m}$ as
\begin{equation*}
\ket{\varphi_m}=\sum_{k=0}^n \sqrt{\binom{n}{k}} \left(\frac{m}{n}\right)^{k/2} \left(1-\frac{m}{n}\right)^{\frac{n-k}{2}} \ket{k}.
\end{equation*}
\end{definition}

The distribution $\lvert \braket{k|\varphi_m}\rvert^2$ describes a binomial law $\mathcal{B}(n,m/n)$, and so the mean and variance are
\begin{equation*}
\braket{\varphi_m|K|\varphi_m}
= m
\quad \text{ and } \quad
\braket{\varphi_m|(K-m)^2|\varphi_m}
= \frac{m(n-m)}{n},
\end{equation*}
where $K$ denotes the linear potential, $K \ket{k} = k \ket{k}$.

Incidentally, as discussed for instance in~\cite{reichardt2004quantum}, the binomial state is also the ground state of the linear Hamming weight Hamiltonian.
\begin{lemma}[(Linear) Hamming weight Hamiltonian] \label{lem:linear-HW}
Fix $m \in [0,n]$ and let
\begin{equation} \label{eq:alpha-m}
\alpha=\frac{n-2m}{\sqrt{m(n-m)}}.
\end{equation}
Then the Hamiltonian $H_\alpha = -A_{hc} + \alpha K$ has ground state $\ket{\varphi_m}$, ground state energy $\lambda_\alpha = \frac{n}{2} \big( \alpha - \sqrt{4+\alpha^2} \big)$ and spectral gap
\[
\Delta_\alpha
= \sqrt{4 + \alpha^2}.
\]
\end{lemma}

\subsection{Discrete log-concavity} \label{sec:discrete-LC}

A key concept in our work is that of (discrete) log-concavity.  In this section we recall some properties of log-concave sequences.

\begin{definition}[Discrete log-concavity~\cite{JOHNSON2007791}]\label{def:1d-logconc}
A non-negative sequence $\{p_k\}_{k\in\mathbb{Z}}$ is log-concave on interval $I \subseteq \mathbb{Z}$ if
\begin{equation} \label{eq:LC}
p_k^2\ge p_{k+1}p_{k-1}
\end{equation}
for all $k \in I$.
If $I = \mathbb{Z}$ then we simply say that $\{p_k\}_{k\in\mathbb{Z}}$ is log-concave.
\end{definition}

\noindent
Equivalently, defining the score sequence $\{r_k\}$ by $r_k=p_{k+1}/p_k$, log-concavity at $k$ is equivalent to $\{r_k\}$ being non-increasing from $k-1$ to $k$.
It turns out that many classical discrete distributions belong to this class: the discrete 
uniform, Bernoulli, binomial, Poisson, geometric, and negative binomial 
distributions, as well as convolutions of Bernoulli distributions with 
arbitrary parameters (see~\cite{Johnson_2006} and references therein).

In this list, the binomial distribution will play a privileged role throughout this work.
In fact, one can base a strengthened notion of log-concavity on the binomial distribution.

\begin{definition}[{Discrete ultra-log-concavity \cite{liggett1997ultra}}] \label{def:ULC}
Fix integer $n \geq 0$ and let $q_k = \binom{n}{k}$ for $k \in [n]$.
A non-negative sequence $\{p_k\}_{k\in [n]}$ is ultra-log-concave (of order $n$) if
\begin{equation} \label{eq:ULC}
\left(\frac{p_k}{q_k}\right)^2
\ge \frac{p_{k+1}}{q_{k+1}} \frac{p_{k-1}}{q_{k-1}} \qquad (0 < k < n).
\end{equation}
\end{definition}

\subsubsection{Useful facts and bounds}
 
We will use the following two elementary facts.

\begin{lemma}
Fix an interval $I \subseteq \mathbb{Z}$ and a non-negative sequence $\{p_k\}_{k \in \mathbb{Z}}$.\\
$\bullet$ If $\{p_k\}$ is concave on $I$, then it is log-concave on $I$. \\
$\bullet$ If $\{p_k\}$ is log-concave on $I$, then so is $\{c^k p_k\}$ for any $c \geq 0$.
\end{lemma}
\begin{proof}
For the first bullet we combine concavity ($2p_k\geq {p_{k+1}+p_{k-1}}$) with the AM-GM inequality (${p_{k-1}+p_{k+1}}\geq 2\sqrt{p_{k-1} p_{k+1}} $) to yield
\[
p_k^2
\geq \big(\frac{p_{k+1}+p_{k-1}}{2} \big)^2
\geq p_{k-1}p_{k+1}.
\]
The second bullet trivially follows from the fact that the rescaling by $c^k$ drops out in the log-concavity condition \cref{eq:LC}.
\end{proof}

\begin{remark} \label{remark:hypercube-coeff}
A consequence of the first bullet is that the hypercube coefficients $a_k=\sqrt{k(n+1-k)}$ from \cref{eq:sym-red} are log-concave, since they are concave:
\[
a_{k+1}+a_{k-1}-2a_k
=\sqrt{(k+1)(n-k)}+\sqrt{(k-1)(n+2-k)}-2\sqrt{k(n+1-k)}<0
\]
\end{remark}

We continue with some crucial yet less elementary claims.
A first key property of a log-concave distribution is that it is \emph{unimodal}, meaning that there exists a mode $k^*$ such that $p_{k-1} \leq p_k$ for $k \leq k^*$ and $p_k \geq p_{k+1}$ for $k \geq k^*$, and so in particular $\max_k p_k = p_{k^*}$.
\begin{claim}[{{\cite[Theorem 3]{Keilson01061971}}}]
A log-concave distribution $\{p_k\}$ is unimodal.
\end{claim}

Moreover, the mode of a log-concave distribution can be related to its variance.

\begin{claim}[{{\cite[Theorem 1.1]{bobkov2021concentrationfunctionsentropybounds}}}]
\label{eq:LC-bound-1}
Let $\sigma^2$ be the variance of a log-concave distribution $\{p_k\}$.
Then
\begin{equation*}
\frac{1}{\sqrt{1+12\sigma^2}}\le p_{k^*}\le\frac{2}{\sqrt{1+4\sigma^2}}.
\end{equation*}
\end{claim}
Another important property of log-concave distributions is that they exhibit geometric decay away from their mode. This is captured by the following tail bound.
\begin{claim}[{{\cite[Lemma 2.3]{marsiglietti2024notestatisticaldistancesdiscrete}}}]
\label{eq:LC-bound-2}
If $\mu$ is the mean of a log-concave distribution $\{p_k\}$, then for all $\alpha\ge 0$ it holds that
\begin{equation*} \sum_{|k-\mu|\ge \alpha} p_k \;\le\; 2 \exp\!\left( -\frac{\alpha}{2\bigl(2\sum_{k} |k-\mu|\,p_k + 1\bigr)} \right). \end{equation*} 
\end{claim}

\begin{remark}
From this bound we can derive the more convenient bound
\begin{equation*}
\sum_{|k-\mu|\geq\alpha} p_k
\le 2\exp\!\left(-\frac{\alpha}{2(2\sigma+1)}\right).
\end{equation*}
This follows from Jensen's inequality applied to the concave function $\sqrt{\cdot}$:
\[
\sum_k |k-\mu| p_k
= \sum_k \sqrt{(k-\mu)^2}p_k
\leq \sqrt{\sum_k (k-\mu)^2 p_k}
= \sigma.
\]
\end{remark}

\subsubsection{Overlapping log-concave distributions}

We will prove a final bound that relates the mean and variance of two log-concave distributions that have a large overlap.

\begin{lemma} \label{lem:log-conc-overlap}
Let $\{p_k\},\{q_k\}$ be two log-concave distributions on $\mathbb{Z}$ with means $\mu_p$ and $\mu_q$ and variances $\sigma_p^2$ and~$\sigma_q^2$ such that $\left(\sum_k \sqrt{p_k q_k}\right)^2 \ge \frac12$.
Then
\[
\mu_q = \mu_p + O(\sigma_p) + O(1), \qquad
\sigma_q = \Theta(\sigma_p) + O(1).
\]
\end{lemma} 
\begin{proof}
We first show that for any $I \subset \mathbb{Z}$,
\begin{equation} \label{eq:qI}
q(I) \ge \left(\frac{1}{\sqrt{2}} - \sqrt{p(I^c)}\right)^2.
\end{equation}
Indeed, by Cauchy--Schwarz,
\[
\sum_{k \in I} \sqrt{p_k q_k}
\le \sqrt{\sum_{k\in I} p_k}\sqrt{\sum_{k\in I} q_k}
= \sqrt{p(I)\,q(I)}
\le \sqrt{q(I)},
\]
and similarly $\sum_{k \notin I} \sqrt{p_k q_k} \le \sqrt{p(I^c)\,q(I^c)} \le \sqrt{p(I^c)}$.
Since $|\braket{p|q}| \ge 1/\sqrt{2}$, this yields
\begin{equation*}
\frac{1}{\sqrt{2}}
\le \sum_k \sqrt{p_k q_k}
= \sum_{k \in I} \sqrt{p_k q_k}
+ \sum_{k \notin I} \sqrt{p_k q_k}
\le \sqrt{q(I)} + \sqrt{p(I^c)},
\end{equation*}
which implies \cref{eq:qI}.

Now let $I_p = [\mu_p - 10\sigma_p,\; \mu_p + 10\sigma_p] \cap \mathbb{Z}.$
By Chebyshev's inequality, $p(I_p^c) \le \frac{1}{100}$.
By \cref{eq:qI} this implies that
$$
q(I_p)
\ge \left(\frac{1}{\sqrt{2}} - \frac{1}{10}\right)^2
\ge \frac{9}{25}.
$$
Letting $q_{\max} = \max_{k\in\mathbb{Z}} q(k)$ we then get that
\begin{equation*}
q_{\max}
\ge \frac{1}{|I_p|} \sum_{k\in I_p} q(k)
= \frac{q(I_p)}{|I_p|}
\ge \frac{9}{25}\frac{1}{1+20\sigma_p}.
\end{equation*}
By log-concavity and \cref{eq:LC-bound-1} we also have
$$
q_{\max} \le \frac{2}{\sqrt{1+4\sigma_q^2}}.
$$
Combining both inequalities gives
$\sqrt{1+4\sigma_q^2}
\le \frac{25}{18}(1+20\sigma_p),$
and hence $\sigma_q \le \frac{25}{36}(1+20\sigma_p).$
By symmetry (interchanging $p$ and $q$) we also obtain
$\sigma_p \le \frac{25}{36}(1+20\sigma_q).$
Together these imply that $\sigma_q = \Theta(\sigma_p) + O(1)$.

Finally, define $I_q = [\mu_q - 10\sigma_q,\; \mu_q + 10\sigma_q] \cap \mathbb{Z}$.
Chebyshev's inequality gives $q(I_q)\ge 99/100$. Since $q(I_p)\ge 9/25$, the intervals $I_p$ and $I_q$ must overlap. Hence
\begin{align*}
\mu_q - 10\sigma_q &< \mu_p + 10\sigma_p
\qquad \text{ and } \qquad
\mu_q + 10\sigma_q > \mu_p - 10\sigma_p,
\end{align*}
which implies $\mu_p - 10(\sigma_p+\sigma_q) < \mu_q < \mu_p + 10(\sigma_p+\sigma_q)$.
Using the previous bounds on the variances then yields $\mu_q = \mu_p + O(\sigma_p) + O(1)$.
\end{proof}

\section{(Ultra-)Log-concavity for convex potentials (and beyond)} \label{sec:lc_gs}

As our first contribution, we establish a large family of potentials $V$ for which the corresponding ground state is discretely log-concave or even ultra-log-concave.
We do so for 1-dimensional stoquastic Hamiltonians of the form
\[
H
= -A + V
\]
where $V$ encodes the potential on $\{0,1,...,n\}$, and $A$ is defined by $\braket{k-1|A|k} = \braket{k|A|k-1} = a_k$ for a concave, positive sequence $\{a_k \in \mathbb{R}\}_{k=1}^n$.
When $a_k = \sqrt{k(n+1-k)}$ this corresponds to the symmetry-reduced adjacency matrix $A_{hc}$ of the hypercube from \cref{sec:sym}.

\subsection{Prior work} \label{sec:continuous}

In the continuous setting it was established by Brascamp and Lieb~\cite[Theorem 6.1]{brascamp1976extensions} that the ground state of a continuous
Schrödinger operator $-\partial^2 + V$ on the real line with a convex potential is log-concave.
Stronger even, they showed that if $V(x) = \alpha^2 x^2 + W(x)$ for a convex $W$, then the ground state $\psi$ has the form $\psi(x) = e^{-\alpha x^2} \varphi(x)$ for some log-concave function $\varphi$.
Clearly, this is a continuous analogue to discrete ultra-log-concavity (\cref{def:ULC}).
Apart from being a useful structural property, this was also used later to establish strong bounds on the spectral gap of these operators~\cite{spectral_gap_continuous, Yu1986LowerBO}.

Transferring this picture to the discrete setting turns out to be non-trivial.
Indeed, a clear bottleneck when transferring the proof in \cite{brascamp1976extensions} to the discrete setting is that discrete log-concavity loses some of the nice properties of its continuous variant.
E.g., in the continuous setting, multivariate 
log-concavity is well-understood and enjoys strong stability properties: in 
particular, \cite{brascamp1976extensions} uses that marginalization of a log-concave distribution remains 
log-concave.
In the discrete setting, however, 
defining a meaningful notion of multivariate log-concavity is more delicate, 
and key properties such as preservation under marginalization fail \cite[Chapter 6]{barndorff2014information}.

One exception that we found is the work by Johnson \cite{JOHNSON2007791}.
As explained in e.g.~\cite{gozlan2023log}, Johnson proves a discrete version of the Brascamp-Lieb result, but only for one particular $A$ matrix on the half-line~$\mathbb{N}$ that corresponds to a particular queuing process on $\mathbb{N}$.

\subsection{Ultra-log-concavity from convexity} \label{sec:ULC-convex}

As our first result in this direction, we establish a discrete version of the aforementioned result by Brascamp and Lieb \cite{brascamp1976extensions}.
Let $\ket{\psi}=\sum_{k=0}^n \psi_k \ket{k}$ denote the ground state of $H=-A+V$ with energy $E_0$, and define the function
\[
F(k):=V(k)-E_0.
\] 
By Perron--Frobenius and positivity of the $a_k$'s, we can assume $\psi_k > 0$ for all $0 \leq k \leq n$.
Projecting the eigenvalue equation $H\ket{\psi}=E_0\ket{\psi}$ onto $\bra{k}$ gives a key recurrence relation:
\begin{equation}\label{eq:recurrence}
F(k)\psi_k=a_{k+1}\psi_{k+1}+a_k\psi_{k-1},
\qquad 0\le k\le n,
\end{equation}
with the boundary convention $\psi_{-1}=\psi_{n+1}=0$.
Given that the $a_k$'s and $\psi_k$'s are positive, this implies that also $F(k) > 0$ for all $0 \leq k \leq n$.
We use this recurrence relation to prove our theorem.

\begin{theorem}[Ultra-log-concavity for convex potentials] \label{thm:ultra log-conc-conv}
If $\{a_k = \sqrt{k(n+1-k)}\}$ and $V$ is convex, then $\{\psi^2_k\}$ is ultra-log-concave.
\end{theorem}

\begin{proof}
We will prove that $\phi_k=\psi_k/\sqrt{\binom{n}{k}}$ is log-concave, which is equivalent to proving that the ratios $r_k=\frac{\phi_k}{\phi_{k-1}}$ for $1 \leq k \leq n$ are non-increasing.
By contradiction, we will suppose $r_k < r_{k+1}$ for some $k$.

We start by rewriting \cref{eq:recurrence} as 
\[
F(k)
= (n-k) r_{k+1} + k/r_k.
\]
Now let $a \geq 0$ be the first index such that $r_a < r_{a+1}$, and let $n \geq b\geq a+1$ be the first index after $a$ such that $r_b > r_{b+1}$, so we have
\[r_{a-1} \geq r_a <r_{a+1} \le \ldots \le r_{b-1} \le r_{b} > r_{b+1},\]
where we formally set $r_0 = \infty$ and $r_{n+1} = 0$.
Using this, we can lower bound $F(a)-F(a-1)$ and upper bound $F(b)-F(b-1)$ by
\[
    F(a)-F(a-1) > \frac{1}{r_a}-r_a,
    \qquad
    F(b)-F(b-1) < \frac{1}{r_{b}}-r_{b}.
\]
Since the function $r \mapsto \frac{1}{r}-r$ is decreasing for positive $r$, this implies that 
\[
F(a)-F(a-1)
> F(b)-F(b-1).
\]
However, our convexity assumption on $V$ implies convexity of $F$, and so this contradicts our assumption.
 \end{proof}

\subsection{\texorpdfstring{$\gamma$}{gamma}-log-concavity from \texorpdfstring{$\gamma$}{gamma}-relative monotonicity} \label{sec:LC-hill}

As our second contribution in this direction, we describe a more technical framework for establishing log-concavity of the ground state.
Critically, it allows us to go beyond convex potentials, showing that log-concavity can persist even in the occurrence of local minima.
While somewhat technically involved, we show that the framework is tight in (and motivated by) capturing limiting cases such as the uniform and binomial ground state.

In the most general case, we will focus on proving a weighted log-concavity property of the ground-state, i.e. 
\begin{equation} \label{eq:gLC}
\psi_k^2
\ge \gamma_k \, \psi_{k-1} \psi_{k+1}
\end{equation}
where the weights $\{\gamma_k\}$ are defined such that:
\[
\gamma_1,\ldots,\gamma_{n-1}>0,
\qquad
\gamma_0=\gamma_n=1.
\]
More precisely, we will say that $\{\psi_k\}$ is $\gamma$-log-concave ($\gamma$-LC) on an interval $I \subseteq \{0,\dots,n\}$ if \cref{eq:gLC} is satisfied for all $k \in I$.
As an example, ultra-log-concavity of $\{\psi_k^2\}$ (\cref{def:ULC}) corresponds to the case where $\gamma_k = \sqrt{\frac{k+1}{k} \frac{n-k+1}{n-k}}$.

To prove $\gamma$-LC, we assume that $\{a_k\}_{k=1}^n$ is positive (while setting $a_0=a_{n+1}=0$) and satisfies a related log-concavity condition:
\begin{equation}\label{eq:ak_cond}
    a_k^2 \gamma_k \gamma_{k-1}\ge a_{k-1}a_{k+1},
\qquad \forall 1\le k\le n.
\end{equation}
Recalling the function $F(k) = V(k) - E_0 > 0$, we also need to define the cumulative weights
\[
\Gamma_{k}^-:=\frac{a_1}{F(0)}\prod_{\ell=0}^{k-1}\gamma_\ell,
\qquad
\Gamma_{k}^+:=\frac{a_n}{F(n)}\prod_{\ell=k+1}^{n}\gamma_\ell,
\]
with the convention that empty products are equal to $1$.
These weights are motivated by the following bound.
\begin{lemma}[Cumulative ratio] \label{lem:cumulative}
\, \newline
$\bullet$ 
If $\{\psi_\ell\}$ is $\gamma$-LC on $\{0,\dots,k-1\}$ then $\frac{\psi_{k-1}}{\psi_k} \geq \Gamma^-_k$.
\, \newline
$\bullet$ 
If $\{\psi_\ell\}$ is $\gamma$-LC on $\{k+1,\dots,n\}$ then $\frac{\psi_{k+1}}{\psi_k} \geq \Gamma^+_k$.
\end{lemma}
\begin{proof}
The first bullet follows from propagating $\gamma$-LC from $\ell = k-1$ to $\ell = 1$:
\[
\frac{\psi_{k-1}}{\psi_k}
\ge \gamma_{k-1} \frac{\psi_{k-2}}{\psi_{k-1}}
\ge \cdots \ge \frac{\psi_0}{\psi_1}\prod_{l=1}^{k-1}\gamma_l
\overset{\text{\eqref{eq:recurrence}}}{=} \frac{a_1}{F(0)} \prod_{l=1}^{k-1}\gamma_l
= \Gamma_{k}^-.
\]
The second bullet follows from propagating $\gamma$-LC from $\ell=k+1$ to $\ell = n-1$:
\[
\frac{\psi_{k+1}}{\psi_k}
\ge \gamma_{k+1}\frac{\psi_{k+2}}{\psi_{k+1}}
\ge \cdots
\ge \frac{\psi_n}{\psi_{n-1}} \prod_{l=k+1}^{n-1}\gamma_l
\overset{\text{\eqref{eq:recurrence}}}{=}
\frac{a_n}{F(n)} \prod_{l=k+1}^{n-1}\gamma_l
= \Gamma_{k}^+. \qedhere
\]
\end{proof}

We can now define a condition linking the potential $V$ to $\{a_k\}$ and $\{\gamma_k\}$ under which we will prove $\gamma$-LC of $\{\psi_k\}$.

\begin{definition}[$\gamma$-relative monotonicity (\,$\gamma$-RM\,)] \label{def:gamma-hill}
The Hamiltonian $H=-A+V$ satisfies the left $\gamma$-RM condition on $\{0,\dots,k_0\}$ if, for all $1 \leq k \leq k_0$,
\begin{equation} \label{eq:leftRM}
a_k\gamma_k F(k)-a_{k+1}F(k-1)
\le
\Gamma_{k-1}^-\left(a_k^2\gamma_k\gamma_{k-1}-a_{k+1}a_{k-1}\right).
\end{equation}
It satisfies the right $\gamma$-RM condition on $\{k_1,\dots,n-1\}$ if, for all $k_1 \leq k \leq n-1$,
\begin{equation} \label{eq:rightRM}
a_{k+1}\gamma_k F(k)-a_kF(k+1)
\le
\Gamma_{k+1}^+\left(a_{k+1}^2\gamma_{k+1}\gamma_k-a_{k+2}a_k\right).
\end{equation}
\end{definition}

Under this condition, we can prove weighted log-concavity of $\{\psi_k\}$ on the interval $I$.

\begin{theorem}[$\gamma$-LC from $\gamma$-RM]
\label{informal:thm_lc_hill}
Assume that $\{a_k\}$ satisfies \cref{eq:ak_cond}. \\
$\bullet$
If $H = -A + V$ satisfies the left $\gamma$-RM condition on $\{1,\dots,k_0\}$ then $\{\psi_k\}$ is $\gamma$-LC on $\{0,\dots,k_0\}$. \\
$\bullet$
If $H = -A + V$ satisfies the right $\gamma$-RM condition on $\{k_1,\dots,n-1\}$ then $\{\psi_k\}$ is $\gamma$-LC on $\{k_1,\dots,n\}$.
\end{theorem}

In particular, if these conditions hold for some $k_1 = k_0 + 1$ then $\{\psi_k\}$ is $\gamma$-LC on $[n]$.

\subsection{Proofs of \texorpdfstring{$\gamma$}{gamma}-LC from \texorpdfstring{$\gamma$}{gamma}-RM}

The proofs of \cref{informal:thm_lc_hill} are by strong induction.

\begin{proof}[Strong induction from the left]
The boundary initialization at $k = 0$ is trivial because we set $\psi_{-1} = 0$.
Now let $k \le k_0$ and assume $\gamma$-LC for all $\ell \in \{0,\ldots,k-1\}$.
By \cref{lem:cumulative} we have
\begin{equation} \label{eq:ratio-bound}
\frac{\psi_{k-1}}{\psi_k}
\geq \Gamma_{k}^-.
\end{equation}
To prove $\gamma$-LC at $k$, we start from the recurrence relation, \cref{eq:recurrence}.
After multiplying by $a_{k+1} \gamma_{k-1} \psi_k$ this is
\[
a_{k+1}\gamma_{k-1}F(k-1)\psi_k \psi_{k-1}=a_ka_{k+1}\gamma_{k-1}\psi_k^2+a_{k-1}a_{k+1}\gamma_{k-1}\psi_k\psi_{k-2}.
\]
By left $\gamma$-RM at $k$ and $\gamma$-LC at $k-1$ this implies that
\[
a_k \gamma_k \gamma_{k-1}F(k) \psi_k \psi_{k-1}-\gamma_{k-1}\Gamma_{k-1}^-\left(a_k^2\gamma_k\gamma_{k-1}-a_{k+1}a_{k-1}\right)\psi_k \psi_{k-1}\leq a_ka_{k+1}\gamma_{k-1}\psi_k^2+a_{k-1}a_{k+1}\psi_{k-1}^2.
\]
Again using \cref{eq:recurrence} and absorbing $\gamma_{k-1}$ this is equivalent to
\[
a_k a_{k+1}\gamma_k \gamma_{k-1} \psi_{k-1}\psi_{k+1}+a_k^2 \gamma_k \gamma_{k-1}\psi_{k-1}^2-\Gamma_{k}^-\left(a_k^2\gamma_k\gamma_{k-1}-a_{k+1}a_{k-1}\right)\psi_k \psi_{k-1}
\leq a_ka_{k+1}\gamma_{k-1}\psi_k^2+a_{k-1}a_{k+1}\psi_{k-1}^2.
\]
Regrouping yields
\[
\left(a_k^2\gamma_k\gamma_{k-1}-a_{k+1}a_{k-1}\right)\psi_k \psi_{k-1} \left( \frac{\psi_{k-1}}{\psi_k}-\Gamma_k^-\right) \leq a_ka_{k+1}\gamma_{k-1}(\psi_k^2-\gamma_k \psi_{k-1}\psi_{k+1}).
\]
By our condition on $\{a_k\}$ and \cref{eq:ratio-bound} we have positivity of the left hand side, and so
\[
\gamma_k \psi_{k-1}\psi_{k+1} \leq \psi_k^2. \qedhere
\]
\end{proof}

\begin{proof}[Strong induction from the right.]
We proceed in an almost identical manner.
The boundary initialization at $k = n$ is trivial because $\psi_{n+1} = 0$.
Now let $k \geq k_1$ and assume $\gamma$-LC for all $\ell \in \{k+1,\ldots,n\}$.
By \cref{lem:cumulative} we have
\begin{equation}  \label{eq:ratio-bound-2}
\frac{\psi_{k+1}}{\psi_k}
\ge \Gamma_{k}^+.
\end{equation}
To prove $\gamma$-LC at $k$, we start from the recurrence relation, \cref{eq:recurrence}.
After multiplying by $a_k \gamma_{k+1} \psi_k$ this is
\[
a_{k}\gamma_{k+1}F(k+1)\psi_k \psi_{k+1}
= a_ka_{k+1}\gamma_{k+1}\psi_k^2+a_{k}a_{k+2}\gamma_{k+1}\psi_k\psi_{k+2}
\]
By right $\gamma$-RM at $k$ and $\gamma$-LC at $k+1$ this implies that
\[
a_{k+1} \gamma_k \gamma_{k+1}F(k) \psi_k \psi_{k+1} - \gamma_{k+1}\Gamma_{k+1}^ + \left(a_{k+1}^2\gamma_k\gamma_{k+1} - a_{k}a_{k+2}\right) \psi_k \psi_{k+1}
\leq a_ka_{k+1}\gamma_{k+1}\psi_k^2 + a_{k}a_{k+2}\psi_{k+1}^2.
\]
Again using \cref{eq:recurrence} and absorbing $\gamma_{k+1}$ this is equivalent to
\[
a_k a_{k+1}\gamma_k \gamma_{k+1} \psi_{k-1}\psi_{k+1}+a_{k+1}^2 \gamma_k \gamma_{k+1}\psi_{k+1}^2-\Gamma_{k}^+\left(a_{k+1}^2\gamma_k\gamma_{k+1}-a_{k}a_{k+2}\right)\psi_k \psi_{k+1}
\leq a_ka_{k+1}\gamma_{k+1}\psi_k^2+a_{k}a_{k+2}\psi_{k+1}^2.
\]
Regrouping yields
\[
\left( a_{k+1}^2\gamma_k\gamma_{k+1}-a_{k}a_{k+2} \right) \psi_k \psi_{k+1} \left( \frac{\psi_{k+1}}{\psi_k} - \Gamma_k^+ \right)
\leq a_k a_{k+1} \gamma_{k+1} (\psi_k^2 - \gamma_k \psi_{k-1} \psi_{k+1}).
\]
By our condition on $\{a_k\}$ and \cref{eq:ratio-bound-2} we have positivity of the left hand side, and so
\[
\gamma_k \psi_{k-1}\psi_{k+1} \leq \psi_k^2. \qedhere
\]
\end{proof}

\subsection{Examples of \texorpdfstring{$\gamma$}{gamma}-LC from \texorpdfstring{$\gamma$}{gamma}-RM} \label{sec:LC-concrete}

To better grasp the full reach of this framework, we make explicit some particularly interesting cases.

\subsubsection{Uniform \texorpdfstring{$\gamma$}{gamma}, uniform \texorpdfstring{$a_k$}{ak}}

In this case, $\gamma$-LC (cf.~\cref{eq:gLC}) reduces to normal log-concavity, and the $\gamma$-RM conditions (cf.~\cref{def:gamma-hill}) elegantly reduce to 
\[
V(k) \le V(k-1)
\qquad (1 \leq k \leq k_0), \qquad\qquad
V(k) \le V(k+1)
\qquad (k_1 \leq k \leq n-1).
\]
In particular, if $V$ is a single-well potential (equivalently, $-V$ is unimodal) then the ground state of $H$ is log-concave on the full interval $[n]$.

\subsubsection{Uniform \texorpdfstring{$\gamma$}{gamma}, log-concave \texorpdfstring{$a_k$}{ak}} \label{sec:uniform-gamma-lc-a}


For uniform $\{\gamma_k\}$, the condition \cref{eq:ak_cond} reduces to log-concavity of the $\{a_k\}$.
By \cref{remark:hypercube-coeff} this setting includes as a particular case of interest the symmetry-reduced hypercube Hamiltonian, in which case $\{a_k=\sqrt{k(n+1-k)}\}$.
We will see that such non-trivial $\{a_k\}$ can yield log-concave ground states even for potentials with multiple wells.

A simple yet important limit case is when $V(k)=a_{k}+a_{k+1}$ (where we set $a_0 = a_{n+1} = 0$).
This makes the Hamiltonian $H$ correspond to the Laplacian of the weighted path, the ground state of which we know to be the uniform state with energy $E_0=0$.
In some sense this is the limit case of log-concavity, but it is also the limit case of our theorem because it saturates the $\gamma$-RM conditions with equality.

Going beyond this canonical example, we observe that the conditions are satisfied if
\begin{equation} \label{eq:cond-F}
\frac{F(k)}{F(k-1)}
\le \frac{a_{k+1}}{a_k} \qquad (1 \leq k \leq k_0), \qquad\qquad
\frac{F(k+1)}{F(k)}
\ge \frac{a_{k+1}}{a_k} \qquad (k_1 \leq k \leq n-1),
\end{equation}
and so $F$ has to satisfy a sort of relative monotonicity condition as compared to $\{a_k\}$ (hence the term $\gamma$-RM).
It is this additional slack that allows us to prove log-concavity of the ground state for numerous non-monotone potentials.
The following lemma captures one such family, of which we display an instance in \cref{fig:nonconvex}.

\begin{lemma}
If $\{a_k=\sqrt{k(n+1-k)}\}$ and $\langle \varphi_{n/2}|V|\varphi_{n/2}\rangle \le n$, the conditions
\begin{align*}
\frac{V(k)}{V(k-1)}
\leq \frac{a_{k+1}}{a_k}
\qquad (1 \le k \le n/2), \qquad\qquad
\frac{V(k+1)}{V(k)}
\geq \frac{a_{k+1}}{a_k}
\qquad (n/2 < k \leq n),
\end{align*}
are sufficient for the corresponding $H$ to satisfy \cref{eq:cond-F} and hence to have a log-concave ground state.
\end{lemma}
\begin{remark}
    The state $|\varphi_{n/2}\rangle$ is the binomial state with mean at $n/2$ (see Section \ref{ssec:binHW}). The condition on the expected value of $V$ at $n/2$ is implied e.g.~when $\max_k{V(k)} \le n$. In the extreme case, $V$ increases from $k=0$ to $k=n/2$. This only imposes that $V(0) \le 2\sqrt{n}$ because $\max V(k) = V(n/2)\le \frac{a_{n/2+1}}{a_1}V(0) \simeq \frac{\sqrt{n}}{2}V(0)$. We get the same argument on the right side, i.e. $V(n)\le 2\sqrt{n}$.
\end{remark}
\begin{proof}
For $k\leq n/2$, we have $1-\frac{a_{k+1}}{a_k}\leq 0$.
Since $E_0\leq\langle \varphi_{n/2}|H|\varphi_{n/2}\rangle=-n+\langle \varphi_{n/2}|V|\varphi_{n/2}\rangle \le 0$, it follows that $E_0\big(1-\frac{a_{k+1}}{a_k}\big)\ge 0$.
Combined with our assumption on $V$ this yields
\[
V(k)-V(k-1) \frac{a_{k+1}}{a_k}
\leq 0
\leq E_0 \left(1-\frac{a_{k+1}}{a_k}\right),
\]
and this implies the first inequality in \cref{eq:cond-F}.
Similarly, for $k\geq n/2$, we have $(1-\frac{a_{k+1}}{a_k})\geq 0$, and therefore  $E_0\big(1-\frac{a_{k+1}}{a_k}\big)\leq 0$.
Combined with our assumption on $V$ this yields
\begin{equation*}
V(k+1)-V(k)\frac{a_{k+1}}{a_k}
\geq 0
\geq E_0 \left(1-\frac{a_{k+1}}{a_k}\right),
\end{equation*}
and this implies the second inequality in \cref{eq:cond-F}.
\end{proof}

\begin{figure}[h]
\centering
\includegraphics[width=\linewidth]{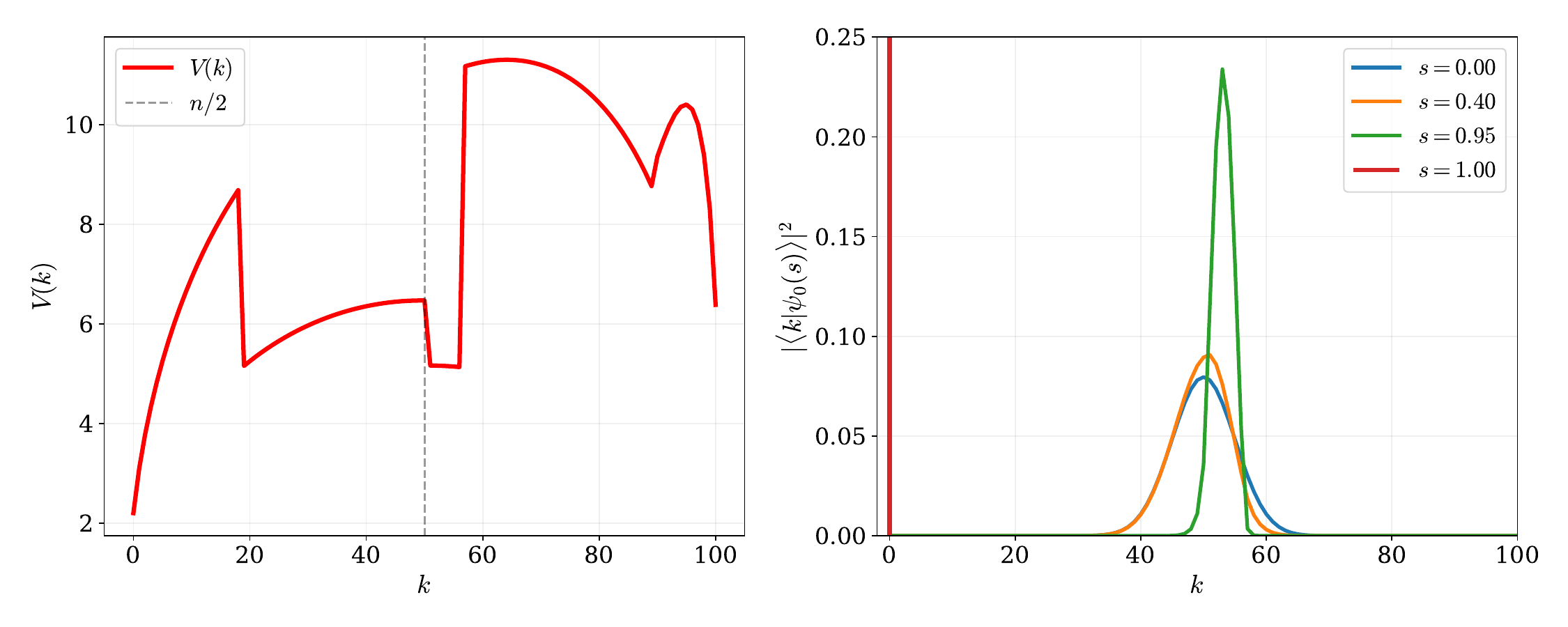}
\caption{Example of a non-monotone potential satisfying the RM condition for 
$n=100$ with $a_k=\sqrt{k(n+1-k)}$ and $\max V(k) \leq n$. Panel~(a) shows the potential, which satisfies the one-sided ratio bounds 
$\frac{V(k)}{V(k-1)}\leq \frac{a_{k+1}}{a_k}$ for $k\leq n/2$ and 
$\frac{V(k+1)}{V(k)}\geq \frac{a_{k+1}}{a_k}$ for $k\geq n/2+1$, where $E_0$ is the ground-state energy of $H = -A+V$.
Panel~(b) shows the shape of 
the log-concave ground state of $H(s) = -(1-s)A + sV$ for different values of~$s$.}
\label{fig:nonconvex}
\end{figure}

\subsubsection{Binomial \texorpdfstring{$\gamma$}{gamma}, binomial \texorpdfstring{$a_k$}{ak}}

As a final example of our framework, we consider the case where $\left\{\gamma_k = \sqrt{\frac{(k+1)(n-k+1)}{k(n-k)}} \right\}$.
In this case, $\gamma$-LC corresponds to ultra-log-concavity (ULC) of $\{\psi_k^2\}$ (\cref{def:ULC}).

We will argue that our framework again tightly captures the limiting case, which is the binomial distribution.
Indeed, in \cref{ssec:binHW} we already saw that the squared ground state of the Hamming weight Hamiltonian (with linear potential $V(k) = \alpha k$) is a binomial distribution.
In the following we show that this example satisfies our $\gamma$-RM conditions with equality.

We first observe that $\gamma_k\frac{a_k}{a_{k+1}}=1+\frac{1}{n-k}$ and $\gamma_k\frac{a_{k+1}}{a_{k}}=1+\frac{1}{k}$, so the $\{a_k\}$ satisfy \cref{eq:ak_cond}.
We then use a telescoping product to rewrite the product 
\[
\prod_{\ell=1}^{k-1} \gamma_\ell
= \sqrt{\frac{kn}{n-k+1}}, \qquad
\prod_{\ell=k+1}^{n-1} \gamma_\ell 
= \sqrt{\frac{n(n-k)}{k+1}}.
\]
This simplifies the $\gamma$-RM conditions to
\begin{align*}
\left( n-k\right)[V(k) - V(k-1)]+V(k)-E_0
&\leq \frac{n^2}{V(0)-E_0}
\qquad \qquad (1 \leq k  \leq k_0), \\
k[V(k) - V(k+1)]+V(k)-E_0
&\leq \frac{n^2}{V(n)-E_0}
\qquad \qquad (k_1 \leq k \leq n-1).
\end{align*}
Now, for the linear potential $V(k) = \alpha k$ we have the exact expression $E_0 = \frac{n}{2}(\alpha-\sqrt{\alpha^2+4})$ (cf.~\cref{lem:linear-HW}), and this is easily checked to satisfy both inequalities with equality.

\section{Spectral gap bounds from (ultra-)log-concavity} \label{sub:spectral_gap_lc}

To demonstrate the power of log-concavity, we show that it implies
useful lower bounds on the spectral gap $\Delta$ of $H$.
The argument builds on a discrete Hardy inequality described by Miclo~\cite{miclo1999example}.

\subsection{Hardy inequality}

We again consider a 1-dimensional Hamiltonian $H=-A+V$
on $\{0,\dots,n\}$ with diagonal potential $V$ and $A$ band-diagonal with $\braket{k-1|A|k} = \braket{k|A|k-1} = a_k$.
We will assume that $\{a_k\}$ is positive, in which case there is a unique ground state $\ket{\psi_0}=\sum_{\ell=0}^n \psi_\ell\ket{\ell}$ that can be chosen positive.
We denote by~$\Delta > 0$ the spectral gap of $H$.

We can connect such a stoquastic Hamiltonian to a Markov chain through the so-called ground state transformation.
Indeed, the operator
\[
L=-\text{Diag}(\psi_\ell)^{-1}(H-E_0)\text{Diag}(\psi_\ell)
\]
is a valid Markov chain generator, satisfying $\braket{k|L|\ell} \geq 0$ for $k \neq \ell$ and $L \vec{1} = 0$.
More precisely, it corresponds to a birth-and-death process on $\{0,\dots,n\}$ with stationary probabilities
\[
p_k
= \psi_k^2
\]
and ergodic flows
\[
b_\ell
= p_{\ell-1} L_{\ell-1,\ell}
= p_\ell L_{\ell,\ell-1}
= a_\ell\sqrt{p_\ell p_{\ell-1}}.
\]

Following \cite{miclo1999example}, we associate to $L$ a family of discrete Hardy constants: for any $0 < i < n$,
\[
B_+(p,i)
:=\max_{i<k\le n} p([k,n])\sum_{\ell=i+1}^k \frac1{b_\ell},
\qquad\quad
B_-(p,i)
:=\max_{0\le k<i} p([0,k])\sum_{\ell=k+1}^i \frac1{b_\ell}.
\]
The following lemma shows that these Hardy constants tightly characterize the spectral gap of $H$.

\begin{lemma}[Discrete Hardy inequality]\label{lem:hardy} For any index $0 < i < n$,
\[
\frac{1}{4}\min \left( \frac{1}{B_+(p,i)}, \, \frac{1}{B_-(p,i)} \right)\le \Delta \le \min \left( \frac{1}{p([0,i])B_+(p,i)}, \, \frac{1}{p([i,n])B_-(p,i)} \right).
\]
\end{lemma}
\begin{proof}
The lower bound follows from the lower bound in \cite[Proposition 3]{miclo1999example}.
The upper bound is not stated explicitly as such, but it follows from the proof of \cite[Proposition 2]{miclo1999example}.
Indeed, replacing the median~$m$ by $i$ in that proof recovers the bound that we claim.
\end{proof}

\begin{remark}
A particularly relevant case is when $i=m$ a median of $p$, so that $p([0,m]) \geq 1/2$ and $p([m,n]) \geq 1/2$.
Then
\[
\frac{1}{4 B}
\leq \Delta
\leq \frac{2}{B}
\]
with $B = \max\{B_+(p,m),B_-(p,m)\}$.
This corresponds to \cite[Proposition 2]{miclo1999example}.
\end{remark}

\subsubsection{Hardy constant for uniform and binomial distribution}

We first treat some cases of special interest.
Consider the uniform distribution over $\{0, \dots,n\}$.

\begin{lemma}[Uniform Hardy constant] \label{lem:unif_B}
Let $S_n=\sum_{l=1}^n \frac{1}{a_l}$ and $\{p_k = \frac{1}{n+1}\}$.
Then, for any $i$,
\[
\max\{B_+(p,i),B_-(p,i)\}
\le (n+1)S_n.
\]
\end{lemma}
\begin{proof}
Explicitly writing out the constants yields
\[
B_+(p,i)
= \max_{i<k\le n} \frac{n-k+1}{n+1}\sum_{l=i+1}^k \frac{n+1}{a_l}
\le (n-i+1)S_n
\]
and
\[
B_-(p,i)
= \max_{0\le k < i} \frac{k+1}{n+1} \sum_{l=k+1}^i \frac{n+1}{a_l}
\le (i+1)S_n. \qedhere
\]
\end{proof}

We will also need a bound on the Hardy 
constants for the binomial distribution $p_k \propto c^k \binom{n}{k}$ with $\{a_k=\sqrt{k(n+1-k)}\}$ derived from the hypercube symmetry reduction.
Rather than deriving it explicitly, we can use that the corresponding Hamiltonian is precisely the symmetry reduction of the Hamming weight Hamiltonian $-A + \alpha K$ as described in \cref{ssec:binHW}, for the particular choice of $\alpha=\frac{n-2\mu}{\sqrt{\mu(n-\mu)}}$ (with $\mu$ the mean of~$q$).
From \cref{lem:linear-HW} we know that this Hamiltonian has a spectral gap $\Delta = \sqrt{\alpha^2+4}\ge 2$.
Combined with the upper bound on the gap from \cref{lem:hardy}, we get the lemma below.

\begin{lemma}[Binomial Hardy constant] \label{lem:bin_B}
Let $a_k=\sqrt{k(n-k+1)}$, $p_k \propto c^k \binom{n}{k}$, and $m_p$ a median of $p$ satisfying $p([0,m_p]) \geq 1/2$ and $p([m_p,n]) \geq 1/2$.
\begin{itemize}
    \item If $i\ge m_p$, then $B_+(p,i) \le 1$.
    \item If $i\le m_p$, then $B_-(p,i) \le 1$.
\end{itemize}
\end{lemma}

\subsection{Hardy constant from relative log-concavity}

In this section we show one way of bounding the Hardy constant: if a distribution $p$ is ``log-concave relative to'' some other distribution $q$,  then we can upper bound the Hardy constant of $p$ using that of~$q$.
For distributions $p$ and $q$, we say that $p$ is log-concave relative to $q$ if the distribution $\{p_k/q_k\}$ is log-concave.
It is easily seen that any log-concave distribution $p$ is relative log-concave with respect to $q_k \propto c^k$ (\cref{def:1d-logconc}).
More interestingly, if $q_k\propto \binom{n}{k}$ then we recover the notion of ultra-log-concavity (\cref{def:ULC}).
We will be using the following technical lemma.

\begin{lemma}[Single crossing lemma] \label{lem:single-crossing}
Suppose that $p$ is log-concave relative to $q$, and let $k^*$ denote the mode of $\{r_k = p_k/q_k\}$.
Let $b_\ell = a_\ell \sqrt{p_{\ell-1} p_\ell}$ and $c_\ell = a_\ell \sqrt{q_{\ell-1} q_\ell}$.
\begin{enumerate}
\item
For all $k,\ell$ such that $k^* < \ell \leq k$, we have
\[
p([k,n])
\leq r_k q([k,n])
\qquad \text{ and } \qquad
b_\ell \geq r_k c_\ell.
\]
\item
Assume there exists $0 < i < k^*$ such that either $q([0,i]) = p([0,i])$ or $q([i,n]) = p([i,n])$.
Then for all $k,\ell$ such that $k > i$ and $i < \ell \leq k^*$, we have
\[
p([k,n]) \leq q([k,n])
\qquad \text{ and } \qquad
b_\ell \geq c_\ell.
\]
\end{enumerate}
\end{lemma}
\begin{proof}
To prove the first point, we note that log-concavity (and hence unimodality) of $\{r_k\}$ implies that the sequence is non-increasing for $k \geq k^*$.
This implies that
\[
p([k, n])
= \sum_{j=k}^n r_j q_j \le r_k \sum_{j=k}^n q_j
= r_k q([k, n]),
\]
but also that $p_\ell \geq r_k q_\ell$ for $k^* \leq \ell \leq k$, and hence $b_\ell \geq r_k c_\ell$ for $k^* < \ell \leq k$.

For the second point, we assume that $p([0,s]) = q([0,s])$ for either $s = i$ or $s = i-1$ with $i < k^*$.
Equivalently, this is $\sum_{j=0}^s (r_j - 1) q_j = 0$.
Since the sequence $\{r_j\}$ is non-decreasing on $\{0,\dots,k^*\}$, this implies that $r_s \geq 1$, and moreover $1 \leq r_s \leq r_{s+1} \leq \dots \leq r_{k^*}$.
This already implies that $b_\ell \geq c_\ell$ for $i < \ell \leq k^*$.

By unimodality and the fact that $r_i \geq 1$, the sequence $r_j - 1$ on $\{i,\dots,n\}$ must be nonnegative up to some point, and nonpositive past that point (hence ``single crossing'').
By our assumption, the weighted sum over the interval $\sum_{j=i+1}^n (r_j - 1) q_j$ is nonpositive.\footnote{It is zero if $q([0,i]) = p([0,i])$. It is nonpositive if $q([i,n]) = p([i,n])$, which follows from the fact that $r_i \geq 1$.}
This implies that any suffix $\sum_{j=k}^n (r_j - 1) q_j$ of the sum  must be nonpositive for any $k > i$, and hence $p([k,n]) \leq q([k,n])$ for $k > i$.
\end{proof}

The following theorem then shows how to bound the Hardy constant of a distribution $p$ by using the Hardy constant of a distribution $q$ relative to which it is log-concave and with which it shares the same tail mass.
Note that we can always enforce this condition by replacing $\{q_k\}$ with $\{c^k q_k\}$ for an appropriate choice of $c$ (indeed, log-concavity relative to $\{q_k\}$ implies log-concavity relative to~$\{c^k q_k\}$).

\begin{theorem}[Relative Hardy bound] \label{thm:relative_lc}
Suppose that $p$ is log-concave relative to $q$, and let $k^*$ be the mode of $\{p_k/q_k\}$.
Let $0 < i < n$ be such that either
\[
q([0,i]) = p([0,i])
\qquad \text{ or } \qquad
q([i,n]) = p([i,n])
\qquad \text{ or } \qquad
i = k^*.
\]
Then
\[
B_+(p,i) \leq B_+(q,i)
\qquad \text{ and } \qquad
B_-(p,i) \leq B_-(q,i).
\]
\end{theorem}
\begin{proof}
We only prove the bound on $B_+$, since the bound on $B_-$ follows by symmetry.

Let $b_\ell = a_\ell \sqrt{p_{\ell-1} p_\ell}$ and $c_\ell = a_\ell \sqrt{q_{\ell-1} q_\ell}$.
The bound on $B_+$ follows if we prove
\[
\frac{p([k,n])}{b_\ell}
\leq \frac{q([k,n])}{c_\ell}
\quad (i < \ell \leq k).
\]
This follows from point 1.~in \cref{lem:single-crossing} if either $k^* \leq i$, or $k^* > i$ and $k^* < \ell \leq k$.
The remaining case where $k^* > i$ and $i < \ell \leq k^*$ follows from point 2.~in \cref{lem:single-crossing}, where we invoke the assumption that either $q([0,i]) = p([0,i])$ or $q([i,n]) = p([i,n])$.
\end{proof}

\subsection{Spectral gap bounds} \label{sec:sp-gap-bounds}

Combined with our criteria for log-concavity of the ground state, we get our main result of the section.

\begin{theorem}\label{thm:lc-gap} 
Consider 1-dimensional $H = -A+V$ with ground state distribution $p$ and spectral gap $\Delta$.
\begin{enumerate}
\item
\textbf{Uniform, log-concave:}
if $\{a_\ell = 1\}$ and $p$ log-concave then
\[
\Delta \in \Omega(1/n^2).
\]
This holds e.g.~when $\{a_\ell = 1\}$ and $V$ is single-well (\cref{sec:LC-concrete}).
\item
\textbf{Hypercube, log-concave:}
if $\{a_\ell=\sqrt{\ell(n+1-\ell)}\}$ and $p$ log-concave then
\[
\Delta \in \Omega(1/n).
\]
This holds e.g.~when $\{a_\ell=\sqrt{\ell(n+1-\ell)}\}$ and $H$ satisfies the 1-RM condition (\cref{sec:LC-concrete}).
\item 
\textbf{Hypercube, ultra-log-concave:}
if $\{a_\ell=\sqrt{\ell(n+1-\ell)}\}$ and $p$ ultra-log-concave then
\[
\Delta \in \Omega(1).
\]
This holds e.g.~when $\{a_\ell=\sqrt{\ell(n+1-\ell)}\}$ and $V$ is convex (\cref{thm:ultra log-conc-conv}).
\end{enumerate}
\end{theorem} 
\begin{proof}
All proofs go via the discrete Hardy inequality (\cref{lem:hardy}), stating that $\Delta \geq \frac{1}{4B(p,i)}$ with $B(p,i) = \max\{B_+(p,i),B_-(p,i)\}$.
To bound $B(p,i)$, we combine our relative Hardy bound (\cref{thm:relative_lc}) with our bounds on the Hardy constants of the geometric (\cref{lem:unif_B}) and binomial distribution (\cref{lem:bin_B}).

Towards proving points 1. and 2., we note that if $p$ is log-concave with mode $k^*$ then it is log-concave relative to the uniform distribution $\{ q_k \propto 1\}$ and the mode of $p/q$ is at $k^*$. We apply \cref{thm:relative_lc} with $i=k^*$.
Hence
\begin{equation} \label{eq:Bp-bound}
B(p,i)
\overset{\text{\cref{thm:relative_lc}}}{\leq} B(q,i)
\overset{\text{\cref{lem:unif_B}}}{\leq} S_n (n+1).
\end{equation}
Point 1.~then follows by noting that $S_n = n$ in that case.
Point 2.~follows by showing that in this case $S_n \le  \pi$ for every $n \geq 1$.
Define
\[
f(x)=\frac{1}{\sqrt{x(n+1-x)}} \qquad (0<x<n+1)
\]
so that $f(\ell) = 1/a_\ell$.
Since \(f(x)=f(n+1-x)\), $f$ is symmetric about the midpoint. Also, \(f\) is positive and decreasing on
\(\left(0,\frac{n+1}{2}\right]\), so for $k\le \frac{n+1}{2}$, $f(k) \le \int_{k-1}^k f(x) dx$. Hence we may bound the sum by an integral:
\[
S_n \le 2 \sum_{k=1}^{\lfloor(n+1)/2\rfloor} f(k)\le 2\int_0^{(n+1)/2} \frac{dx}{\sqrt{x(n+1-x)}}.
\]
Now make the substitution $x=(n+1)\sin^2\theta$ so that $dx = 2(n+1)\sin\theta\cos\theta\,d\theta$ and
\[
\sqrt{x(n+1-x)}
= \sqrt{(n+1)\sin^2\theta \cdot (n+1)\cos^2\theta}
= (n+1)\sin\theta\cos\theta.
\]
Therefore
\[
\int_0^{(n+1)/2} \frac{dx}{\sqrt{x(n+1-x)}}
=
\int_0^{\pi/4} 2\,d\theta
=
\frac{\pi}{2}
\]
and so $S_n \le 2\cdot \frac{\pi}{2}=\pi$.

Finally we prove point 3., which corresponds to the case where $p$ is log-concave relative to the binomial distribution $q$. We apply \cref{thm:relative_lc} with $i=m_p$ a median of $p$, satisfying $p([0,m_p]) \geq 1/2$ and $p([m_p,n]) \geq 1/2$.

\noindent If $B(p,m_p)=B_+(p,m_p)$, then we choose $c$ such that $q([0,m_p])=p([0,m_p])$.
Then if $m_q$ is the leftmost median of $q$, we have that $m_q \le m_p$.
Using our bound on the Hardy constant of any binomial (\cref{lem:bin_B}), we get
\[
B(p,m_p)=B_+(p,m_p)
\overset{\text{\cref{thm:relative_lc}}}{\leq} B_+(q,m_p)
\overset{\text{\cref{lem:bin_B}}}{\leq} 1.
\]
If $B(p,m_p)=B_-(p,m_p)$, then we choose $c$ such that $q([m_p,n])=p([m_p,n])$.
Letting $m'_q$ denote the rightmost median of $q$, we get that $m'_q \ge m_p$.
Using our bound on the Hardy constant of any binomial (\cref{lem:bin_B}), we get 
\[
B(p,m_p)=B_-(p,m_p)
\overset{\text{\cref{thm:relative_lc}}}{\leq} B_-(q,m_p)
\overset{\text{\cref{lem:bin_B}}}{\leq} 1. \qedhere
\]
\end{proof}

\noindent
It is interesting to compare our bounds with the results by Jarret and Jordan in \cite{Jarret_2014}:
\begin{itemize}
\item
For the 1-dimensional Hamiltonian $H = -A + V$ with uniform $\{a_\ell\}$ and $V$ convex, they show a bound $\Omega(1/n^2)$ on the spectral gap.
Point 1.~in \cref{thm:lc-gap} generalizes this bound to the much weaker case where $-V$ is merely unimodal.
\item
For the hypercube Hamiltonian $H = -A_{hc} + V$ with a symmetric potential $V$, we can use the symmetry reduction in \cref{sec:sym}.
Point 3.~in \cref{thm:lc-gap} then shows a constant gap when $V$ is convex, in correspondence to the findings in \cite{Jarret_2014}.
Point 2.~in \cref{thm:lc-gap} goes beyond this, and shows an $\Omega(1/n)$ lower bound when $H$ satisfies the $1$-RM condition (\cref{def:gamma-hill}).
This includes potentials that are non-convex and can even have local minima.
\end{itemize}
\noindent
We do mention that the bounds in Jarret and Jordan \cite{Jarret_2014} are tight down to the constant, in analogy to the continuous fundamental gap proof \cite{spectral_gap_continuous} which their work is based on.
This implies that points 1.~and 3.~in  \cref{thm:lc-gap} are tight up to a constant.

Regarding point 2.~in \cref{thm:lc-gap}, we note that this bound is also tight up to a constant.
Indeed, for the aforementioned potential $V(k)=a_k +a_{k+1}$, the ground state is uniform and hence log-concave.
If the $a_k$'s are derived from the hypercube then point 2.~in \cref{thm:lc-gap} gives a bound $\Omega(1/n)$.
To see that this is tight, note that the ground state energy in this case is 0, and the first excited energy can be shown to be $O(1/n)$ by using an ansatz state proportional to $\sum_k (k-n/2)|k\rangle$, orthogonal to the uniform ground state.

\section{Log-concavity and tunneling} \label{sec:LC-tunnel}

We now demonstrate the utility of log-concavity in analyzing quantum tunneling through localized spike perturbations.
Following a variation of Reichardt's argument for the HWS problem~\cite{reichardt2004quantum}, we show that for log-concave ground states we can easily control the spectral gap when adding a spike perturbation.
In the following we define the mean $\mu$ and variance $\sigma^2$ of a ground state $\ket{\psi_0}$ as the mean and variance of the distribution $p_k = \psi_k^2$.
The key point will be that for a log-concave ground state, the mass placed on a short interval is uniformly controlled by the variance and the distance of the interval from the mean.

\begin{theorem}\label{thm:spike_general}
Consider a Hamiltonian $H = -A + V$ with spectral gap $\Delta$ and a log-concave ground state with variance $\sigma$ and mean $\mu$.
Consider a perturbation $h:\{0,1,\dots,n\} \to \mathbb{R}_{\geq 0}$ supported on an interval $[\ell,u-1] \subseteq [0,n]$ and satisfying an upper bound $h(k) \leq \bar h$.
If $\alpha = \max\{\mu-u,\ell-\mu\}$ and
\[
\bar{h} \cdot \min\left\{ \frac{2(u-\ell)}{\sqrt{1+4\sigma^2}}, 2\exp\left(-\frac{\alpha}{2(2\sigma+1)}\right) \right\}
\le \Delta/2
\]
then the perturbed Hamiltonian $H' = H + V'$ with $V' \ket{k} = h(k) \ket{k}$ has spectral gap $\Delta' \geq \Delta/2$.
\end{theorem}
\begin{proof}
Let $\lambda_0,\lambda_1$ (resp.~$\lambda_0',\lambda_1'$) denote the two lowest eigenvalues of $H$ (resp.~$H'$).
By Weyl’s monotonicity theorem, we have that $\lambda'_1 \ge \lambda_1$.
Conversely, with $p_k = \psi_k^2$, the variational principle with $\ket{\psi_0}$ as a test function implies that
\begin{equation*}
\lambda_0' - \lambda_0
\le \braket{\psi_0|V'|\psi_0}
\leq \bar{h} \sum_{k=\ell}^{u-1} p_k.
\end{equation*}
Since the $p_k$'s are log-concave, we know from \cref{eq:LC-bound-1} that
\[
p_{k^*} \leq \frac{2}{\sqrt{1+4\sigma^2}}.
\]
This yields a bound $\lambda_0' - \lambda_0 \leq \bar{h} \sum_{k=\ell}^{u-1} p_k \le \bar{h} (u-\ell) p_{k^*} \le \bar{h} \frac{2(u-\ell)}{\sqrt{1+4\sigma^2}}$, which is good when $\sigma$ is sufficiently large.

We get an alternative bound when $\sigma$ is sufficiently small as compared to the distance between the mean~$\mu$ and the spike.
Here we use the second bound based on log-concavity, \cref{eq:LC-bound-2}, which states that
\[
\sum_{|k-\mu| \geq \alpha} p_k
\leq 2 \exp\left(\frac{-\alpha}{2(2\sigma+1)}\right)
\]
for any $\alpha \geq 0$.
If either $\mu \leq \ell-\alpha$ or $\mu \geq u+\alpha$, then this implies that
\[
\lambda_0' - \lambda_0 \leq \bar{h} \sum_{k=\ell}^u p_k \le 2 \bar{h} \exp\left(\frac{-\alpha}{2(2\sigma+1)}\right). \qedhere
\]
\end{proof}

\section{Application: Quadratic Hamming weight with a spike}\label{sec:quadra}

In this section we develop the ``quadratic Hamming weight with a spike'' (quadratic HWS) problem.
For this we will specialize our analysis to quadratic potentials perturbed by a spike.
Specifically we consider a family of Hamiltonians
\[
H_{qs}(s)
= -(1-s)A_{hc} + s V_{qs}
\]
for $s \in [0,1]$, with $V_{qs} = V_q + h$ where $V_q(k) = a (k-k_0)^2$ for $k_0 \leq 0$ (ensuring that the minimizer is at $k=0$) and $h$ is the spike perturbation.
We aim to use \cref{thm:spike_general} to bound the spectral gap $\Delta_{qs}(s)$ of $H_{qs}$, which requires a bound on the spectral gap $\Delta_q(s)$ of the \emph{unperturbed} quadratic Hamiltonian $H_q(s) = -(1-s) A_{hc} + s V_q$, as well as estimates of the mean and variance of its ground state $\ket{\psi_q(s)}$.
While for a linear potential, as in the HWS problem, the ground state is known explicitly as a binomial distribution, this is not so for quadratic potentials.
We address this key challenge by showing that $\ket{\psi_q(s)}$ does maintain a constant overlap with the binomial ground state of a linear potential.
Combined with log-concavity of the ground state, this yields explicit control of the mean and variance of~$\ket{\psi_q(s)}$.
This allows us to invoke the perturbation bound from \cref{sec:LC-tunnel} and generalize the HWS problem to quadratic potentials.

\subsection{Unperturbed quadratic: Overlap with the binomial ansatz}\label{sub:overlap}

We now consider the Hamiltonian $H = -A_{hc} + V$ for a convex potential~$V$.
The following claim quantifies the overlap of the ground state of $H$ with an ansatz state $\ket{\varphi_m}$.
The overlap depends on the second difference of the potential, which is defined as $\Delta^2 V(k)=V(k+1)+V(k-1)-2V(k).$ Smaller $\Delta^2 V(k)$ leads to higher overlaps with the binomial state.
We prove the claim in Appendix~\ref{app:overlap_general}.

\begin{restatable}[Overlap bound]{claim}{overlap}
\label{prop:overlap}
Let $\ket{\psi_0}$ denote the ground state of $H = - A_{hc} + V$ with $V$ a convex function with bounded second difference $\Delta^2 V(k)\le C$.
Let $m$ be chosen so as to minimize $\braket{\varphi_m|H|\varphi_m}$.
Then it holds that
\begin{equation*}
|\braket{\psi_0|\varphi_m}|^2
\ge
1-
\left(1-\frac{2m}{n}\right)
\left(
(\mu-m)\left(1-\frac{\Delta V(\lfloor m\rfloor)}{\mathbb{E}[\Delta V(Y)]}\right)
+
\frac{C\,m(1-m/n)}{2\mathbb{E}[\Delta V(Y)]}
\right),
\end{equation*}
where $Y\sim\mathcal{B}(n-1,m/n)$ and $\mu = \braket{\psi_0|K|\psi_0}$.
\end{restatable}

In the quadratic case $V_q(k) = a (k-k_0)^2$, we have $\Delta V_q(k) = a (2(k-k_0)+1)$, and so the lower bound becomes simpler.
We prove the below theorem in Appendix~\ref{app:overlap_quadra}, and we plot the optimal mean~$m^*(a)$ as a function of $a$ in \cref{fig:m_s}.

\begin{theorem}[Overlap with the binomial ansatz for quadratic potentials]
\label{theorem:overlap}
Let $\ket{\psi_{q,a}}$ denote the ground state of $H = -A_{hc} + V_q$ with $V_q(k) = a (k-k_0)^2$ for $k_0 \leq 0$ and $a \in \mathbb{R}_{\geq 0}$.
Then there exists $m^*(a) \in [0,n/2]$ such that
\begin{equation*}
\bigl|\braket{\psi_{q,a}|\varphi_{m^*(a)}}\bigr|^2
\ge \frac{1}{2}.
\end{equation*}
\end{theorem}

\begin{figure}
    \centering
    \includegraphics[width=0.5\linewidth]{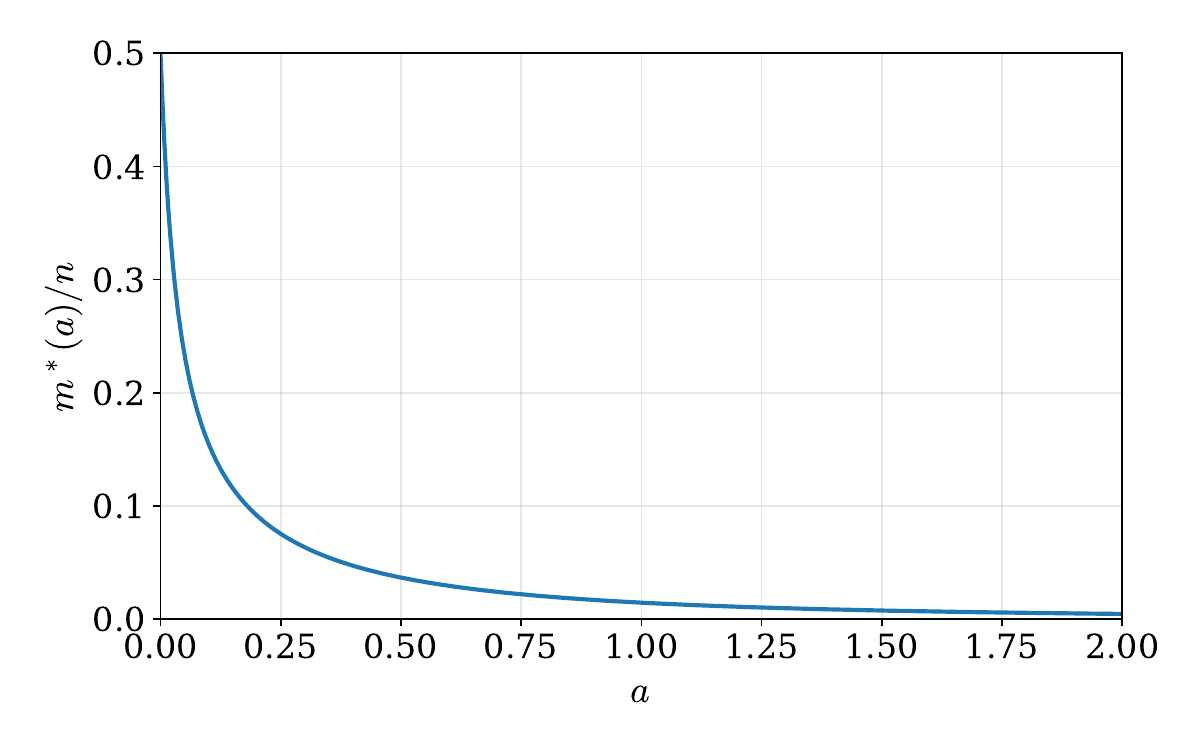}
    \caption{Plot of the function $m^*(a)/n$ obtained for $n=40$ and $k_0=-3$. $m^*(a)$ is obtained by solving the equation $a=\frac{n-2m^*}{\sqrt{m^*(n-m^*)}} \cdot\frac{1}{1 - 2 k_0 + 2(n-1)m^*/n}$. As expected, $m^*$ is located at $n/2$ for $a=0$ and it decreases to $0$ as $a$ increases. }
    \label{fig:m_s}
\end{figure}

Now denote by $\mu_a$ and $\sigma_a^2$ the mean and variance associated to $\ket{\psi_{q,a}}$, which are given by
\begin{equation*}
\mu_a
= \braket{\psi_{q,a}|K|\psi_{q,a}}, \qquad
\sigma_a^2
= \braket{\psi_{q,a}|(K-\mu)^2|\psi_{q,a}}.
\end{equation*}
Using the fact that both $\ket{\psi_{q,a}}$ and $\ket{\varphi_m}$ are log-concave, we can relate the associated mean and variance using \cref{lem:log-conc-overlap}.
Applying this lemma to $\{p_k = |\braket{k|\varphi_{m^*(a)}}|^2\}$ and $\{q_k = |\braket{k|\psi_{q,a}}|^2\}$, we obtain
\begin{align} \label{eq:mean-variance}
\begin{split}
\mu_a &= m^*(a) + O\big(\sqrt{m^*(a)}\big) + O(1),\\
\sigma_a &= \Theta\big(\sqrt{m^*(a)}\big) + O(1).
\end{split}
\end{align}
Thus, although the quadratic ground state does not admit a closed form, the 
binomial ansatz provides a tractable approximation throughout the annealing path.

\subsection{Unperturbed quadratic: Improved spectral gap bound}
\label{sec:impr_gap}

We now turn back to the interpolating, adiabatic Hamiltonian with quadratic potential,
\[
H_q(s)
= -(1-s)A_{hc} + s V_q.
\]
From \cref{thm:lc-gap} (or \cite{Jarret_2014}), we get a bound $\Delta(H_q(s)) \in \Omega(1-s)$ on the spectral gap of the unperturbed, quadratic Hamiltonian $H_q(s)$.
Unfortunately, that bound becomes useless around the end of the annealing path, when $s \to 1$.
Luckily, a finer reading shows that the bound of~\cite{Jarret_2014} can actually be extended until the end of the annealing path.

\begin{lemma} \label{corr:convex_spectral_gap}
Consider $H(s) = -(1-s)A_{hc} + sV$ with $V$ convex, and define $\beta = \min_{0\leq i<j\leq n} \lvert \frac{V(j)-V(i)}{j-i}\rvert$.
Then
\begin{equation*}
\Delta(H(s))\ge\sqrt{s^2\beta^2+4(1-s)^2}.
\end{equation*}
In particular, if there exists $c > 0 $ such that $\beta \geq c$, then
\begin{equation*}
\Delta(H(s)) \ge \frac{2c}{\sqrt{c^2+4}}
\qquad\text{uniformly for all }s\in[0,1].
\end{equation*}
\end{lemma}
\begin{proof}[Proof]
Rewrite
\begin{equation*}
H(s)=(1-s)\widetilde{H}_\lambda,
\qquad
\widetilde{H}_\lambda=-A+\lambda V,
\qquad
\lambda=\frac{s}{1-s},
\end{equation*}
so that $\Delta(H(s))=(1-s)\Delta(\widetilde{H}_\lambda).$
By Lemma~10 of~\cite{Jarret_2014}, for convex $V$ it holds that
\begin{equation*}
\Delta(\widetilde{H}_\lambda)
\ge \Delta(-A+\lambda\beta K)
= \sqrt{(\lambda\beta)^2+4},
\end{equation*}
where $\beta = \min_{i<j} \lvert \frac{V(j)-V(i)}{j-i}\rvert$ and the equality is from \cref{lem:linear-HW}.
This proves that $\Delta(H(s)) \ge (1-s)\sqrt{\left(\frac{s}{1-s}\beta\right)^2+4} = \sqrt{s^2\beta^2+4(1-s)^2}$.

We have $\Delta(H(s))\ge \sqrt{c^2s^2+4(1-s)^2}$.
This is minimized for $s=\frac{4}{4+c^2}$ which then gives $\Delta(H(s)) \geq \frac{2c}{\sqrt{c^2+4}}$ uniformly for $s \in [0,1]$.
\end{proof}

In particular, for $V$ with unique minimizer at $0$, we have $\beta=V(1)-V(0) >0 $ which satisfies by definition this lemma. For the quadratic potential $V_q(k) = a(k-k_0)^2$ that we consider, we have that $\Delta V_q(k)\ge a$.
Assuming $a \in \Omega(1)$, this yields a uniform bound $\Delta_q(s) \in \Omega(1)$ on the spectral gap of the unperturbed quadratic Hamiltonian $H_q(s)$.

\subsection{Perturbed quadratic and AQO}
\label{sub:spike_quadra}

We will now combine the bound on the spectral gap of the quadratic Hamiltonian $H_q(s)$, as well as the estimates on the mean and variance of its ground state, with the perturbative tunneling argument in \cref{sec:LC-tunnel}.
This will allow us to establish our analysis of the quadratic HWS problem.

\subsubsection{Spectral gap of quadratic HWS} \label{sub:spectral_gap_quadratic}

We now combine the general spike bound with the overlap analysis. The key point 
is that the variance of the true ground state is comparable to that of the binomial ansatz.

\begin{theorem}[Quadratic potential with a spike]\label{theorem:gap_quadra}
Consider the quadratic Hamming weight with a spike Hamiltonian
\[
H_{qs}(s)
= - (1-s) A_{hc} + s V_{qs}
= - (1-s) A_{hc} + s (V_q + h)
\]
for $s \in [0,1]$, where the spike $h$ is supported on $[\ell,u-1] \subseteq [0,n]$ and satisfies an upper bound $h(k) \leq \bar h$.
If
\[
\bar{h}\,\frac{u-\ell}{\sqrt{\ell}} \in O(1)
\]
then the spectral gap of $H_{qs}(s)$ satisfies $\Delta_{qs}(s) \in \Omega(1)$.
\end{theorem}
\begin{proof}
Let $\ket{\psi_{q}(s)}$ denote the ground state and  $\Delta_q(s)$ the spectral gap of the unperturbed Hamiltonian $H_q(s) = -(1-s) A_{hc} + s V_q$.
By \cref{corr:convex_spectral_gap} we know that $\Delta_q(s) \in \Omega(1)$.
By Theorem~\ref{thm:spike_general} we get that the perturbed gap $\Delta_{qs}(s) \in \Omega(1)$ provided that
\begin{equation} \label{eq:bound}
\min\left\{ \frac{2(u-\ell)}{\sqrt{1+4\sigma^2(s)}}, 2\exp\left(-\frac{\alpha(s)}{2(2\sigma(s)+1)}\right) \right\}
\leq 1/\bar{h},
\end{equation}
where $\sigma^2(s)$ is the variance of the unperturbed ground state $\ket{\psi_{q}(s)}$ and $\alpha(s) = \max\{\ell-\mu(s),\mu(s)-u\}$ is the distance of the mean $\mu(s)$ of $\ket{\psi_{q}(s)}$ from the support $[\ell,u-1]$ of the spike~$h$.
By applying \cref{theorem:overlap} and \cref{eq:mean-variance} to the Hamiltonian $-A_{hc} + \frac{s}{1-s} V_q$, there exists $m^*(s)$ such that $\mu(s) = m^*(s) + O(\sqrt{m^*(s)}) + O(1)$ and $\sigma(s) = \Theta(\sqrt{m^*(s)}) + O(1)$.
Our proof distinguishes two cases.

First, assume that $m^*(s) \geq \ell/2$.
Then $\sigma(s) \in \Omega(\sqrt{\ell})$ and
\begin{equation*}
\frac{2(u-\ell)}{\sqrt{1+4\sigma^2(s)}}
= \frac{u-\ell}{\sqrt{\ell}}\cdot\frac{2\sqrt{\ell}}{\sqrt{1+4\sigma^2(s)}}
\in O(1/\bar{h})
\end{equation*}
by our assumption that $\bar{h}\,\frac{u-\ell}{\sqrt{\ell}} \in O(1)$.

Then, assume that $m^*(s) < \ell/2$.
Then $\alpha(s) \in \Omega(\ell)$ and $\sigma(s) \in O(\sqrt{\ell}) + O(1)$.
In this case, we get
\[
\exp\left(-\frac{\alpha(s)}{2(2\sigma(s)+1)}\right)
\leq \exp\left(-\Omega\big(\sqrt{\ell}\big)\right).
\]
It remains to note that $\exp\left(-\Omega\big(\sqrt{\ell}\big)\right) \in O(1/\sqrt{\ell})$, and to use that $O(1/\sqrt{\ell}) \in O(1/\bar{h})$ by our assumption.
\end{proof}

The spike condition $\bar{h}\,\frac{u-\ell}{\sqrt{\ell}} \in O(1)$ is identical to the condition derived by Reichardt for the HWS problem~\cite{reichardt2004quantum}.
When for instance $u,\ell \in \Omega(n)$, the condition admits a spike of height and width~$\Omega(n^{1/4})$.

\subsection{Runtime for quadratic HWS} \label{sec:aqo_quadratic_hws}

Finalizing our analysis, we combine our spectral gap bound with the AQO algorithm.
By \cref{sec:AQO}, AQO approximately follows the ground state $\ket{\psi_{qs}(s)}$ of $H_{qs}(s)$ and prepares the target state $\ket{\psi_{qs}(1)} = \ket{0}$ in time 
\begin{equation*}
O\!\left(\frac{(n+\max_k|V(k)|)^2}{(\min_s \Delta(s))^3}\right).
\end{equation*}
This yields the following corollary.

\begin{corollary}[Informal version of Theorem~\ref{theorem:gap_quadra}]
Let $V_q(k)=a(k-k_0)^2$ with $k_0 \leq 0$ and consider the perturbed Hamiltonian 
$H_{qs}(s)=-(1-s)A_{hc}+s(V_q+h)$, where the spike $h$ is supported on $[\ell,u-1] \subseteq [0,n]$ and satisfies $h(k) \leq \bar h$.
If $\bar{h}\,\frac{u-\ell}{\sqrt{\ell}}=O(1)$, then AQO prepares an approximation of $\ket{\psi_{qs}(1)} = \ket{0}$ from $\ket{\psi_{qs}(0)}$, the binomial state with mean at $n/2$, in time~$O(n^4)$.
\end{corollary}

We note that the setup is symmetric with respect to shifting the minimizer $\ket{0}$ to some other, hidden string $\ket{z^*}$, $z^* \in \{0,1\}^n$.
In that case, the potential $V$ would satisfy the symmetry condition $V(x) = V(|x \oplus z^*|)$ and the same analysis goes through.

\section{Summary, related work and open questions}

Our work establishes discrete log-concavity as a critical tool for analyzing discrete optimization with the quantum adiabatic algorithm.
We proved log-concavity for a large family of potentials (including and beyond convex potentials), and used it to establish new spectral gap bounds.
As a concrete application we proved that AQO can efficiently solve the quadratic HWS problem, yielding a direct extension to the results by Reichardt~\cite{reichardt2004quantum} and Farhi, Goldstone and Gutmann \cite{farhi2002quantum} on the (linear) HWS problem.

\subsection{Related work}

In this section we compare our results to some related works.

\subsubsection{Spectral gap analysis of spike Hamiltonians}

There exist more refined analyses of the linear HWS Hamiltonian.
For spike position $u,\ell \approx n/4$, spike width $u-\ell = n^\alpha$ and spike height $h = n^\beta$, it was shown in \cite{reichardt2004quantum} that $\alpha+\beta < 1/2$ implies a constant gap for linear potentials, and our result extends this bound to quadratic potentials.
In later works, Brady and van Dam \cite{brady2016spectral} (improving on Kong and Crosson \cite{kong2015performancequantumadiabaticalgorithm}) showed that if $\alpha + \beta > 1/2$ but $2\alpha + \beta < 1$, then the gap scales as $n^{1/2-\alpha-\beta}$, while for $2\alpha+\beta>1$, the gap becomes exponentially small.
It would be interesting to extend these results to other potentials, as is suggested in \cite{brady2016spectral}, and our result seems to be the first one in this direction.

\subsubsection{Classical benchmark algorithms}

Our polynomial runtime bound for the quadratic HWS problem should be held against comparable classical algorithms (which, similarly to AQO, do not exploit any particular structure of the problem).
To this end, we note that a typical approach based on classical Markov chains and simulated annealing would take exponential time to solve this problem.
This can be demonstrated in an identical fashion to the arguments in \cite{farhi2002quantum,crosson2016simulated} for the linear HWS problem.
Let $h$ be the height and $w = u - \ell$ the width of the spike, and assume $h,w \geq n^c$ for $c \in \Omega(1)$.
If the Markov chain flips less than $w$ bits in its proposal, then it can only cross the spike by accepting a move that jumps into the spike, but this will be rejected with probability $1-O(\exp(-n^c))$.
Conversely, if the Markov chain flips more than $w$ bits, then the probability of a proposal ever jumping to the minimizer is $O(\exp(-n^c))$.

Notably, it was shown by Crosson and Harrow \cite{crosson2016simulated} and Bergamaschi \cite{bergamaschi2020simulated} that a more advanced classical Markov chain algorithm called simulated quantum annealing can solve the (linear) HWS problem under exactly the same assumptions on the spike perturbations.
This algorithm bypasses the former argument by using non-local ``wordline updates''.
It seems reasonable to expect that this algorithm would also solve the quadratic HWS problem, and so we only claim an advantage against the simplest simulated annealing algorithms.

\subsection{Open questions}

Finally, the results of this work suggest many directions for future research.

\paragraph{Sufficient and necessary condition for log-concavity.}
The $\gamma-$RM condition is a sufficient condition.
We leave the question open of whether it is also a necessary condition.
We do prove that at least two important limiting cases of (ultra-)log-concavity (the uniform and the binomial distribution) satisfy the $\gamma-$RM condition with equality.

\paragraph{Can we extend the quadratic HWS problem to general convex potentials?}
We specialized to quadratic potentials with 
minimizer at $0$ because we can prove a strong overlap bound between the corresponding ground states and our binomial ansatz states.
A natural question is whether this approach extends to other convex potentials.
First, note that if the minimizer is not at $0$ then the ansatz is not well suited.
We show this in \cref{fig:overlap} where we plot the overlap with the best binomial ansatz for a shifted quadratic potential $V(k)=(k-k_*)^2$ with minimizer at $k_*=n/3$. 
More unexpectedly, the overlap can decrease even if the minimizer is at $0$.
We show this in the figure for a piecewise-linear potential.

\begin{figure}[h]
    \centering
    \includegraphics[width=0.7\linewidth]{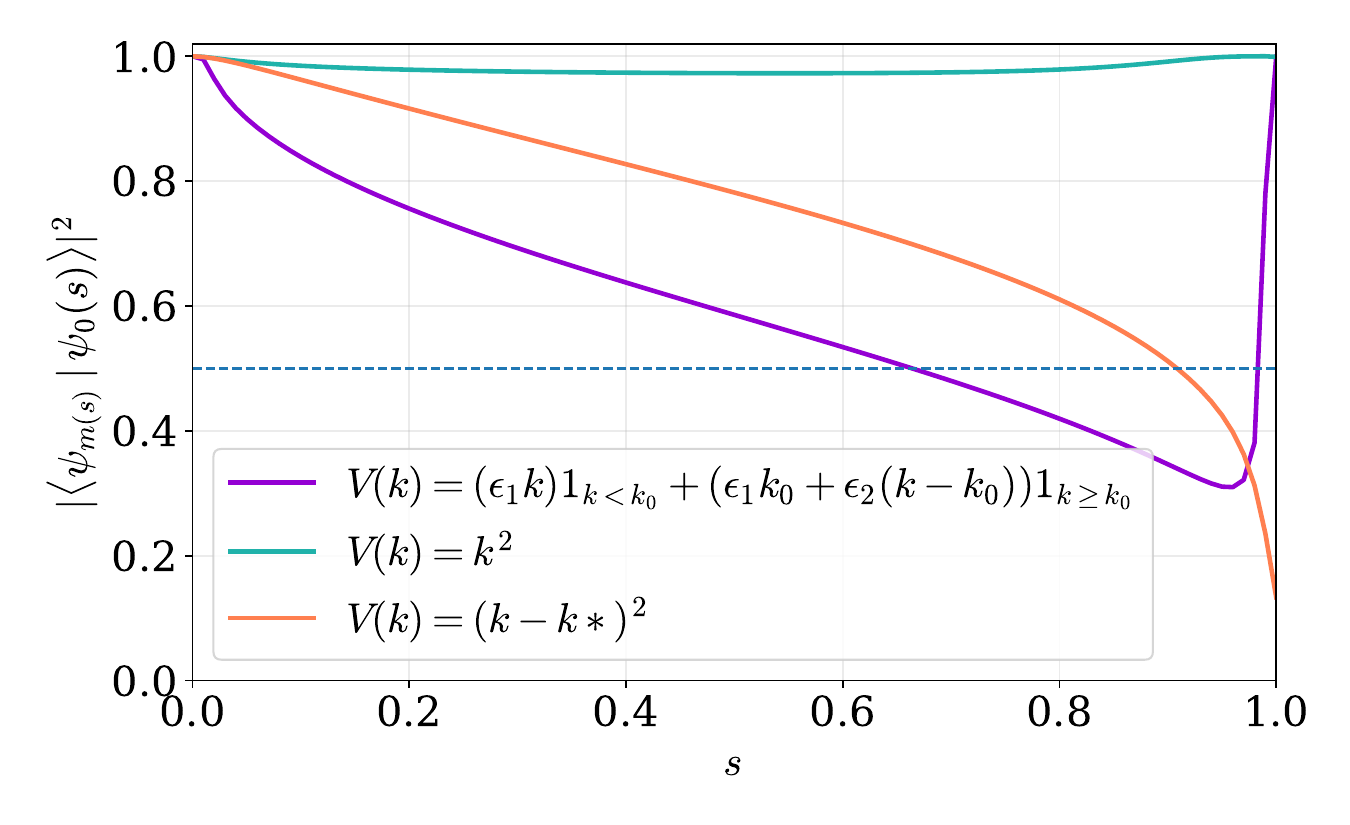}
    \caption{Overlap between the optimal binomial ansatz state $\ket{\varphi_m}$ and the ground state $\ket{\psi_0(s)}$ of $H(s)=-(1-s)A+sV$ for $n=40$ and different potentials: the quadratic potential $V(k)=k^2$ (used in Section~\ref{sec:quadra}), a shifted quadratic potential $V(k)=(k-k_*)^2$ with minimizer at $k_*=n/3$, and a piecewise-linear potential $V(k)=\epsilon_1 k\,\mathbf{1}_{k<k_0}+(\epsilon_1 k_0+\epsilon_2(k-k_0)) \,\mathbf{1}_{k\ge k_0}$, with $n=40$, $\epsilon_1=1/n$, $\epsilon_2=2n$, and $k_0=n/4$.}
    \label{fig:overlap}
\end{figure}

\paragraph{Can this framework be applied to study tunneling for local Hamiltonians?}
A key point is that a localized spike perturbation corresponds to a highly non-local Hamiltonian (coupling many qubits together).
It would be interesting to extend our discussion to local target Hamiltonians.
E.g., the Curie-Weiss model considers target Hamiltonian $H_1=-\sum_{i,j \in [n]} Z_i Z_j$.
As discussed for instance in \cite{Bapst_2012}, the spectral gap of $H(s) = -(1-s) \sum_i X_i + s H_1$ as a function of $s$ undergoes a phase transition from constant to exponentially small in $n$.
Such a phase transition makes the Curie-Weiss model a natural benchmark for our techniques.
We leave it as an open question whether the log-concavity framework developed here can be adapted to handle such local models. 

\paragraph{Can this framework be extended to non-permutation-symmetric 
potentials?}
A key ingredient of our analysis is the reduction from the hypercube to the $(n+1)$-dimensional Hamming weight subspace, which allows us to work on a weighted path graph.
Extending our analysis to more general potentials on the hypercube is a clear open question.
It raises interesting points such as what the correct notions of convexity and log-concavity are on the hypercube.
We expect this to be non-trivial.
E.g., it was shown in~\cite{jarret2015} that single-basin potentials on general graphs, a condition slightly weaker than convexity, can already yield exponentially small spectral gaps.
This contrasts with the 1-dimensional case where monotonicity of the potential on a uniform path was enough (Section \ref{sec:LC-hill}).
Additionally, multivariate log-concavity is less trivial to define in the discrete case than in the continuous case \cite{barndorff2014information,klartag2018poisson}, and identifying the right structural substitute in higher dimensions should be interesting.

Finally, a slightly different approach was taken by Montanaro and Zhou \cite{montanaro2024quantum}, who argued that certain quantum algorithms (QAOA in particular) for symmetric functions can be shown to be robust against bounded, asymmetric perturbations.

\section{Acknowledgements}

SA was supported in part by the European QuantERA project QOPT (ERA-NET Cofund 2022-25) and the French ANR project QUOPS (ANR-22-CE47-0003-01).
SA and AB were supported by the French PEPR integrated project HQI (ANR-22-PNCQ-0002).

\bibliographystyle{alpha}
\bibliography{biblio}

\appendix

\section{Proof of the binomial ansatz overlap} \label{app:overlap_general}

In this appendix, we prove the overlap bound in \cref{prop:overlap} for a convex potential~$V$. 
Throughout, we let $\bar V$ denote the extension of $V$ to the continuous interval $[0,n]$ by linear interpolation, so that
\begin{equation} \label{eq:Vbar}
\bar V(x)
= V(\lfloor x \rfloor) + \{x\} \big( V(\lceil x\rceil) - V(\lfloor x \rfloor) \big),
\end{equation}
where we used the notation $\{x\} = (x - \lfloor x \rfloor)$.

\subsection{Helper lemmas}
We first prove some helper lemmas.
\begin{lemma}[Variational overlap bound] \label{lower_bound:general}
Let $H$ be a Hamiltonian with non-degenerate ground state $\ket{\psi_0}$, ground state energy $\lambda_0$ and spectral gap $\Delta$.
Then for any quantum state $\ket{\phi}$ it holds that
\begin{equation*}
|\braket{\phi|\psi_0}|^2
\ge
1-\frac{\braket{\phi|H|\phi}-\lambda_0}{\Delta}.
\end{equation*}
\end{lemma}
\begin{proof}
Let $(\lambda_k, \ket{\psi_k})_k$ be the eigenvalues and eigenvectors of 
$H$, ordered so that $\lambda_0<\lambda_1 \leq ...$.
Then
\[
\braket{\phi|(H-\lambda_0)|\phi}
= \sum_{k>0}(\lambda_k-\lambda_0)|\braket{\phi|\psi_k}|^2
\ge \Delta\sum_{k>0}|\braket{\phi|\psi_k}|^2
= \Delta\bigl(1-|\braket{\phi|\psi_0}|^2\bigr),
\]
which rearranges to give the result.
\end{proof}

\begin{lemma}[Convex upper bound] \label{lemma:convex_delta}
    Let $V:\mathbb{Z} \to \mathbb{R}$ be a discrete convex function. Then for all $a,b \in \mathbb{R}$, 
    \begin{equation*}
        \bar V(a) - \bar V(b)
        \leq (a-b) \Delta V(\lfloor a \rfloor).
    \end{equation*}
\end{lemma}
\begin{proof}
For the continuous, piecewise differentiable extension $\bar V$ we have the standard fact that
\[
\bar V(a) - \bar V(b)
\leq (a-b) \bar V'(a)
\]
when $\bar V$ is differentiable at $a$.
This happens whenever $a \notin \mathbb{Z}$, in which case $\bar V'(a) = \Delta V(\lfloor a \rfloor)$.
The statement for $a \in \mathbb{Z}$ follows by continuity and taking the limit $\lim_{\varepsilon \to 0_+} a+\varepsilon$.
\end{proof}

\subsection{Binomial expectation lemmas}

Next we prove some technical lemmas on binomial expectations.

\begin{lemma}[Binomial energy] \label{lemma:V_overlap_psi_m}
Let $X \sim \mathcal{B}(n,m/n)$ for $m \in [0,n]$.
If $\Delta^2 V(k)\leq C$ for all $1 \leq k \leq n-1$ then
\begin{equation*}
\mathbb{E}[V(X)]
\leq \bar V(m) +\frac{C}{2}m(1-m/n).
\end{equation*}
\end{lemma}
\begin{proof}[Proof]
Let $W(k) = V(k) - \frac{C}{2} k^2$.
Using that $\Delta^2 \left(\frac{C}{2} k^2 \right) = C$ we get that $\Delta^2 W(k) = \Delta^2 V(k) - C \leq 0$, and so $W$ is discrete concave.
Now let $\bar W$ denote the concave extension of $W$ to the continuous interval $[0,n]$ by linear interpolation.
Jensen's inequality then gives
\[
\mathbb{E}[\bar W(X)]
\leq \bar W(\mathbb{E}[X])
= \bar V(m) - \frac{C}{2} \mathbb{E}[X]^2.
\]
On the other hand,
\[
\mathbb{E}[\bar W(X)]
= \mathbb{E}[\bar V(X)] - \frac{C}{2} \mathbb{E}[X^2].
\]
The claim follows by using that $\mathbb{E}[X^2] - \mathbb{E}[X]^2 = \mathrm{Var}[X] = m(1-m/n)$.
\end{proof}

\begin{lemma}[Binomial derivative] \label{lem:derivative}
Let $m \in [0,n/2]$, $X \sim \mathcal{B}(n,m/n)$ and $Y \sim \mathcal{B}(n-1,m/n)$.
Then
\[
\frac{\mathrm{d}}{\mathrm{d}m} \mathbb{E}[V(X)]
= \mathbb{E}[\Delta V(Y)].
\]
\end{lemma}
\begin{proof}
Let $\mathrm{Pr}[X=k] = q_{m,n}(k) = \binom{n}{k} \left(\frac{m}{n}\right)^k \left(1-\frac{m}{n}\right)^{n-k}$.
Then (e.g.~by the derivative rule of Bernstein polynomials)
\[
\frac{\mathrm{d}}{\mathrm{d}m} q_{m,n}(k)
= q_{m,n-1}(k-1) - q_{m,n-1}(k)
\]
where we set $q_{m,n-1}(-1)=0$.
Plugging in and reshuffling shows that
\[
\frac{\mathrm{d}}{\mathrm{d}m} \sum_{k=0}^n q_{m,n}(k) V(k)
= \sum_{k=0}^n q_{m,n-1}(k) \big( V(k+1)-V(k) \big). \qedhere
\]
\end{proof}

\begin{lemma}[{Optimal $m$}] \label{lem:m}
Let $m\in[0,n/2]$, and choose $m$ such that
\begin{equation*}
\frac{d}{dm}\braket{\varphi_m|-A+V|\varphi_m}=0.
\end{equation*}
Then $\frac{d}{dm}\braket{\varphi_m|V|\varphi_m} = \alpha(m)$ with $\alpha(m) = \frac{n-2m}{\sqrt{m(n-m)}}$ (cf.~\eqref{eq:alpha-m}).
In particular, $m$ satisfies
\begin{equation*}
\alpha(m)
= \mathbb{E}[\Delta V (Y)]
\quad \text{ for } \quad Y \sim \mathcal{B}(n-1,m/n).
\end{equation*}
\end{lemma}
\begin{proof}
The variational energy is 
$-2\sqrt{m(n-m)}+\braket{\varphi_m|V|\varphi_m}$. 
Differentiating and setting to zero,
\begin{equation*}
-\frac{n-2m}{\sqrt{m(n-m)}} + \frac{d}{dm}\braket{\varphi_m|V|\varphi_m}=0.
\end{equation*}
Using $\alpha(m) = (n-2m)/\sqrt{m(n-m)}$ and 
$\frac{d}{dm} \braket{\varphi_m|V|\varphi_m} = \mathbb{E}[\Delta V(Y)]$ (\cref{lem:derivative}) gives the result.
\end{proof}

\subsection{Proof of claim}

Finally we prove the claim, which we restate below.

\overlap*
\begin{proof}[Proof of \cref{prop:overlap}]
Let $\alpha = \alpha(m)$ as in \cref{lem:linear-HW}.
Then $H_\alpha = -A_{hc} + \alpha K$ has ground state $\ket{\varphi_m}$, ground state energy $\lambda_\alpha = \frac{n}{2} \big( \alpha - \sqrt{4+\alpha^2} \big)$ and spectral gap $\Delta_\alpha = \sqrt{4+\alpha^2}$.
Applying \cref{lower_bound:general} with $H=H_\alpha$ and $\ket{\phi}=\ket{\psi_0}$ we get
\begin{equation*} 
|\braket{\psi_0|\varphi_m}|^2
\ge 1-\frac{\braket{\psi_0|H_\alpha|\psi_0}-\lambda_\alpha}{\Delta_\alpha}.
\end{equation*}
We expand the numerator:
\begin{align*}
\braket{\psi_0|H_\alpha|\psi_0}-\lambda_\alpha
&= \braket{\psi_0|H_\alpha|\psi_0} - \braket{\varphi_m|H_\alpha|\varphi_m} \\
&= \alpha(\mu-m) - \bigl(\braket{\psi_0|A_{hc}|\psi_0}-\braket{\varphi_m|A_{hc}|\varphi_m}\bigr).
\end{align*}
Since $\ket{\psi_0}$ is the ground state of $H = -A_{hc}+V$, we have $\braket{\psi_0|H|\psi_0} \leq \braket{\varphi_m|H|\varphi_m}$ and hence
\[
\braket{\psi_0|A_{hc}|\psi_0}-\braket{\varphi_m|A_{hc}|\varphi_m}
\ge \braket{\psi_0|V|\psi_0} - \braket{\varphi_m|V|\varphi_m}.
\]
Therefore
\begin{equation*}
|\braket{\psi_0|\varphi_m}|^2
\ge 1 - \frac{\alpha}{\sqrt{4+\alpha^2}}(\mu-m)
+ \frac{1}{\sqrt{4+\alpha^2}}\bigl(\braket{\psi_0|V|\psi_0} - \braket{\varphi_m|V|\varphi_m}\bigr).
\end{equation*}
By Jensen's inequality, $\braket{\psi_0|V|\psi_0}\ge \bar V(\mu)$ for $\mu = \braket{\psi_0|K|\psi_0}$.
Combining this with Lemma~\ref{lemma:V_overlap_psi_m},
$\braket{\varphi_m|V|\varphi_m}\le \bar V(m)+\frac{C}{2}m(1-m/n)$, yields
\begin{align}
\braket{\psi_0|V|\psi_0} - \braket{\varphi_m|V|\varphi_m}
&\ge \bar V(\mu) - \bar V(m) - \frac{C}{2}m\!\left(1-\frac{m}{n}\right) \label{eq:half-proof} \\
&\ge (\mu-m) \Delta V(\lfloor m \rfloor) - \frac{C}{2}m\!\left(1-\frac{m}{n}\right),
\end{align}
where the last inequality used Lemma~\ref{lemma:convex_delta}, $\bar V(m)-\bar V(\mu)\le(m-\mu)\Delta V(\lfloor m\rfloor)$.

Substituting, and using that $\alpha = \mathbb{E}[\Delta V(Y)]$ from 
Lemma~\ref{lem:m} and $\frac{\alpha}{\sqrt{4+\alpha^2}} = 1 - \frac{2m}{n}$, gives the result.
\end{proof}

\section{Overlap bound for quadratic potentials} 
\label{app:overlap_quadra}

In this appendix, we specialize the overlap bound to the quadratic case, proving \cref{theorem:overlap}.

\begin{claim}\label{claim:early_overlap_claim}
Consider $H = -A_{hc} + V_q$ with $V_q(k)=a(k-k_0)^2$ and ground state $\ket{\psi_0}$.
If $m^*$ minimizes $\braket{\varphi_m|H|\varphi_m}$ then it satisfies
\begin{equation*}
a
= \frac{n-2m^*}{\sqrt{m^*(n-m^*)}}
\cdot\frac{1}{1 - 2m^*/n + 2(m^* - k_0)},
\end{equation*}
and
\[
|\braket{\psi_0|\varphi_{m^*}}|^2
\geq 1/2.
\]
\end{claim}

\begin{proof}
By \cref{lem:m} we have that the optimal $m^*$ satisfies $\frac{n-2m^*}{\sqrt{m^*(n-m^*)}} = \mathbb{E}[\Delta V(Y)]$ with $Y \sim \mathcal{B}(n-1,m^*/n)$.
Since $\Delta V(k)=a(2(k-k_0)+1)$ we can write out the right hand side as
\begin{equation*}
\mathbb{E}[\Delta V(Y)] = a\!\big(1 - 2m^*/n + 2(m^* - k_0)\big),
\end{equation*}
and this yields the first identity.

Towards proving the second bound, we combine this expression with \cref{eq:half-proof} in the proof of \cref{prop:overlap} to get
\[
|\braket{\psi_0|\varphi_m}|^2
\geq 1 - \left(1-\frac{2m}{n}\right) f(\mu,m),
\]
where $\mu = \braket{\psi_0|K|\psi_0}$ and
\[
f(\mu,m)
= \mu - m - \frac{\bar V_q(\mu) - \bar V_q(m) - am\!\left(1-m/n\right)}{a\left(1-2m/n+2(m-k_0)\right)}.
\]
By \cref{lem:V-bar} we can rewrite $\bar V_q(x) = a(x-k_0)^2 + a\{x\}(1-\{x\})$.
Using the bound $0 \leq \{x\}(1-\{x\}) \leq 1/4$, this implies that
\[
\bar V_q(\mu)-\bar V_q(m)
\geq a\left[(\mu-k_0)^2-(m-k_0)^2 - 1/4 \right].
\]
Substituting and combining terms shows that
\begin{align*}
f(\mu, m)
&\leq \mu-m - \frac{(\mu-k_0)^2-(m-k_0)^2 - 1/4 - m(1-m/n)}{1-2m/n+2(m-k_0)} \\
&= \frac{(1-2m/n)(\mu-m) - (\mu-m)^2 + 1/4 + m(1-m/n)}
{1-2m/n+2(m-k_0)} \\
&\leq \frac{(1-2m/n)(\mu-m) - (\mu-m)^2 + 1/4 + m(1-m/n)}
{1-2m/n+2m}.
\end{align*}
We further simplify this bound by noting that
\[
(1-2m/n)(\mu-m) - (\mu-m)^2
\leq 1/4,
\]
which follows from the fact that for $y \in \mathbb{R}$ and $q \in [0,1]$ it holds that $qy - y^2 = -(y-q/2)^2 + q^2/4 \leq q^2/4 \leq 1/4$.
This yields the bound
\[
f(\mu,m)
\leq \frac{1/2 + m(1-m/n)}{1-2m/n+2m}.
\]
It remains to prove that
\[
(1-2m/n) f(\mu,m)
\leq \frac{(1-2m/n)(1/2 + m(1-m/n))}{1-2m/n+2m}
\leq 1/2,
\]
but this is equivalent to the condition $(1-2m/n)(1-m/n) \leq 1$, which is satisfied by our assumption that $m \leq n/2$.
\end{proof}

\begin{lemma} \label{lem:V-bar}
For $V_q(x) = a(x-k_0)^2$ and $x\in [0,n]$ we have
\begin{align*}
\bar V_q(x)
= a(x-k_0)^2 + a\{x\}(1-\{x\}).
\end{align*}
\end{lemma}
\begin{proof}
This follows by using the expressions for $\bar V_q(x)$ (\cref{eq:Vbar}) and $\Delta V_q(\lfloor x \rfloor)$ to write out
\[
\bar V_q(x)
= V_q(\lfloor x \rfloor) + a \{x\} (2(\lfloor x\rfloor-k_0)+1)
\]
and compare it with
\[
a (x-k_0)^2
= a (\lfloor x \rfloor + \{x\} - k_0)^2
= V_q(\lfloor x \rfloor) + a \{x\} (2(\lfloor x\rfloor-k_0)+ \{x\}). \qedhere
\]
\end{proof}

\end{document}